\documentclass[a4paper,UKenglish,cleveref, autoref]{lipics-v2019}

\title{Restricted Max-Min Allocation: Approximation and Integrality Gap} 

\titlerunning{Restricted Max-Min Allocation}

\author{Siu-Wing Cheng}{Department of Computer Science and Engineering, HKUST, Hong Kong}{scheng@cse.ust.hk}{https://orcid.org/0000-0002-3557-9935}{}

\author{Yuchen Mao}{Department of Computer Science and Engineering, HKUST, Hong Kong}{ymaoad@cse.ust.hk}{https://orcid.org/000-0002-1075-344X}{}

\authorrunning{S.\,-W. Cheng and Y. Mao}

\Copyright{Siu-Wing Cheng and Yuchen Mao}

\ccsdesc[100]{Theory of computation~Scheduling algorithms}

\keywords{Fair allocation, configuration LP, approximation, integrality gap}

\category{}

\supplement{}

\acknowledgements{}

\nolinenumbers 

\usepackage{booktabs,tabu}
\usepackage{tikz}
\renewcommand{\leq}{\leqslant}
\renewcommand{\geq}{\geqslant}
\newcommand{\src}{\mathit{src}}
\newcommand{\sink}{\mathit{sink}}

\begin{document}

\maketitle

\begin{abstract}
	Asadpour, Feige, and Saberi proved that the integrality gap of the configuration LP for the restricted max-min allocation problem is at most $4$.  However, their proof does not give a polynomial-time approximation algorithm.   A lot of efforts have been devoted to designing an efficient algorithm whose approximation ratio can match this upper bound for the integrality gap.  In ICALP 2018, we present a $(6 + \delta)$-approximation algorithm where $\delta$ can be any positive constant, and there is still a gap of roughly $2$.  In this paper, we narrow the gap significantly by proposing a $(4+\delta)$-approximation algorithm where $\delta$ can be any positive constant.  The approximation ratio is with respect to the optimal value of the configuration LP, and the running time is $\mathit{poly}(m,n)\cdot n^{\mathit{poly}(\frac{1}{\delta})}$ where $n$ is the number of players and $m$ is the number of resources.  We also improve the upper bound for the integrality gap of the configuration LP to $3 + \frac{21}{26} \approx 3.808$.
\end{abstract}

\section{Introduction}
\label{sec:intro}

\paragraph*{Background} 
In the max-min fair allocation problem, we are given a set $P$ of $n$ players, a set $R$ of $m$ indivisible resources, and a set of non-negative values $\{v_{pr}\}_{p\in P, r\in R}$.  For each $r\in R$ and each $p \in P$, resource $r$ is worth a value of $v_{pr}$ to player $p$.  An allocation is a partition of $R$ into disjoint subsets $\{D_p\}_{p \in P}$ so that each player $p$ is assigned the resources in $D_p$.  The goal is to find an allocation that maximizes the welfare of the least lucky player, that is, we want to maximize
$\min_{p\in P}\sum_{r\in D_p}v_{pr}$. Unfortunately, unless $\mathrm{P}=\mathrm{NP}$, no polynomial-time algorithm can achieve an approximation ratio smaller than 2~\cite{BD05}.

Bez\'{a}kov\'{a} and Dani~\cite{BD05} tried to solve the problem using the assignment LP -- a technique for the classic scheduling problem of makespan minimization~\cite{LST90}.   However, they showed that the integrality gap of the assignment LP is unbounded, so rounding the assignment LP gives no guarantee on the approximation ratio.  Later, Bansal and Sviridenko~\cite{BS06} proposed a stronger LP relaxation, the configuration LP, for the max-min allocation problem.  Asadpour and Saberi~\cite{AS07} developed a polynomial-time rounding scheme for the configuration LP that gives an approximation ratio of $O(\sqrt{n}\log^3 n)$.  Saha and Srinivasan~\cite{SS10} improved it to $O(\sqrt{n\log n})$.  These approximation ratios almost match the lower bound of $\Omega(\sqrt{n})$ for the integrality gap of the configuration LP proved by Bansal and Svirodenko~\cite{BS06}.  Bateni et al.~\cite{BCG09} and Chakrabarty et al.~\cite{CCK09} established a trade-off between the approximation ratio and the running time.  For any $\delta > 0$, they can achieve an approximation ratio of $O(n^{\delta})$ with $O(n^{1/\delta})$ running time.

In this paper, we study the restricted max-min allocation problem.  In the restricted case, we have $v_{pr} \in \{v_r, 0\}$.  That is, each resource $r$ has an intrinsic value $v_r$, and it is worth value $v_r$ to those players who desire it and value $0$ to those who do not.  Assuming $\mathrm{P}\neq \mathrm{NP}$, the restricted case has a lower bound of $2$ for the approximation ratio. The integrality gap of configuration LP for the restricted case also has a lower bound of $2$.  Bansal and Sviridenko~\cite{BS06} proposed an $O\bigl(\frac{\log\log n}{\log\log\log n}\bigr)$-approximation algorithm by rounding the configuration LP.  Feige~\cite{F08} proved that the integrality gap of the configuration LP is bounded by a constant, albeit large and unspecified.  His proof was later made constructive by Haeupler et al.~\cite{HSS11}, and hence a constant approximation can be found in polynomial time.  Asadpour et al.~\cite{AFS12} viewed the restricted max-min allocation problem as a bipartite hyper-graph matching problem.  Let $T^*$ be the optimal value of the configuration LP.  By adapting Haxell's~\cite{H95} alternating tree technique for bipartite hyper-graph matchings, they proposed a local search algorithm that returns an allocation where every player receives at least $T^*/4$ worth of resources, and hence proved that the integrality gap of the configuration LP is at most $4$.   However, their algorithm is not known to run in polynomial time.  A lot of efforts have been devoted to making their algorithm run in polynomial time.  Polacek and Svensson~\cite{PS12} showed that the local search can be done in quasi-polynomial time by building the alternating tree in a more careful way.  
 Annamalai, Kalaitzis and Svensson~\cite{AKS17} carried out the local search in a more structured way.   Together with two new \emph{greedy} and \emph{lazy update} strategies, they can find in polynomial time an allocation in which every player receives a value of at least $T^*/(6+ 2\sqrt{10} + \delta)$.  Recently, we proposed a more flexible, aggressive greedy strategy that 
improves the approximation ratio to $6 + \delta$~\cite{CM18a}.  Davies et al.~\cite{DRZ18} 
claimed a $(6+\delta)$-approximation algorithm for the restricted max-min allocation problem by reducing it to the fractional matroid max-min allocation problem.

\paragraph*{Our Contribution} 
We adapt the framework in~\cite{AKS17} by introducing two new strategies: \emph{layer-level node-disjoint paths} and \emph{limited blocking}.  The performance of our framework is determined by three parameters, and a trade-off between the running time and the quality of solution can be achieved by tuning these parameters.  On one extreme, our framework acts exactly the same as the original local search in~\cite{AFS12}, which achieves a ratio of $4$ but not necessarily run in polynomial time.  On the other extreme, it becomes something like the algorithm in~\cite{AKS17}, which achieves a polynomial running time but a much worse ratio.  We show that, in order to achieve a polynomial running time, one doesn't have to go from one extreme to the other --- a marginal movement is sufficient.  As a result, a ratio slightly worse than $4$ can be achieved in polynomial time.

\begin{theorem}
	\label{thm:approx}
	For any constant $\delta > 0$, there is a $(4+\delta)$-approximation algorithm for the restricted max-min allocation problem that runs in $\mathit{poly}(m,n)\cdot n^{\mathit{poly}(\frac{1}{\delta})}$ time.
\end{theorem}

Although the algorithm we present takes the optimal value of the configuration LP as its input, one can avoid solving the configuration LP by combining our algorithm with binary search to zoom into the optimal value of configuration LP.  The binary search technique is similar to that in~\cite{AKS17, CM18a}.

We also show that the integrality gap of the configuration LP is at most $3 + \frac{21}{26} \approx 3.808$ by giving a better analysis of the AFS algorithm.  This improves the bound of $3+\frac{5}{6} \approx 3.833$ recently obtained in~\cite{CM18b,JR18}.

\begin{theorem}
	\label{thm:gap}
	The integrality gap of the configuration LP for the restricted max-min allocation problem is at most $3 + \frac{21}{26} \approx 3.808$.
\end{theorem}

We focus on only the proof of Theorem~\ref{thm:approx} in the main text and defer the proof of Theorem~\ref{thm:gap} to Appendix~\ref{apd:gap}.
\section{Preliminaries}
\label{sec:pre}

\subsection{The configuration LP}  
Suppose that we hope to find an allocation where every player receives at least $T$ worth of resources.  A \emph{configuration} for a player $p$ is a subset $D$ of the resources desired by $p$ such that $\sum_{r \in D}v_{r} \geq T$.  Let ${\cal C}_p(T)$ denote the set of all configurations for $p$.  

The configuration LP is given on the left of Figure~\ref{fig:lp}.  Given a target $T$, the configuration LP, denoted as $\mathit{CLP}(T)$, associates a variable $x_{p,C}$ with each player $p$ and each configuration $C$ in ${\cal C}_p(T)$. 
Its first constraint ensures that each player receives at least $1$ unit of configurations, and the second constraint guarantees that every resource $r$ is used in at most 1 unit of configurations.  The optimal value of the configuration LP is the largest $T$ for which $\mathit{CLP}(T)$ is feasible. We denote this optimal value by $T^*$.  Without loss of generality, we assume that $T^* = 1$ for the rest of the paper.  Although the configuration LP may have an exponential number of variables, it can be solved within any constant relative error in polynomial time~\cite{BS06}.  Viewing the objective function of the configuration LP as a minimization of a constant, one can get the dual LP on the right of the Figure~\ref{fig:lp}.  

\begin{figure}
		\begin{minipage}[t][11em][t]{0.4\textwidth}
		\centering\textbf{Primal}
		\begin{alignat*}{3}
    	&\quad 	& \sum_{C\in {{\cal C}_p(T)}}x_{p,C} &\geq 1
				& \quad &\forall p\in P\\
		&		&\sum_{p\in P}\sum_{C\in {\cal C}_p(T): r\in C} x_{p, C} &\leq 1
				& \quad &\forall r\in R\\
		&		& x_{p, C}&\geq 0
		\end{alignat*}
		\end{minipage}
		\hfill
		\begin{minipage}[t][11em][t]{0.5\textwidth}
		\centering\textbf{Dual}
		\begin{alignat*}{3}
		\max&\quad & &\sum_{p\in P} y_p - \sum_{r \in R} z_r\\
		s.t.&\quad & y_p&\leq \sum_{r\in C}z_r
			& \quad &\forall p\in P, \forall C\in {\cal C}_p(T)\\
			&		& y_p&\geq 0 & \quad &\forall p\in P\\
			&		& z_r&\geq 0 & \quad &\forall r\in R
		\end{alignat*}
		\end{minipage}
	\caption{The configuration LP and its dual.}
	\label{fig:lp}
\end{figure}

\subsection{Fat and thin edges}

Our goal is to find an allocation in which every player receives at least $\lambda$ worth of resources for some $\lambda \in (0,1)$.  In particular, our approximation algorithm sets $\lambda = \frac{1}{4+\delta}$ where $\delta$ is a positive constant.  For each resource $r\in R$, we call $r$ \emph{fat} if $v_r \geq \lambda$, and \emph{thin} otherwise.  To find the target allocation, it suffices to assign each player $p$ either a fat resource desired by $p$ or a subset $D$ of the thin resources desired by $p$ with $\sum_{r\in D}v_r \geq \lambda$.

For every $p\in P$ and every fat resource $r$ desired by $p$, we call $\{p, r\}$ a \emph{fat edge}.  For every $p \in P$ and every subset $D$ of the thin resources desired by $p$, we call $(p, D)$ a \emph{thin edge} if $\sum_{r \in D}v_r \geq \lambda$.  Two edges are \emph{compatible} if they share no common resource.  We say that a fat edge $\{p,r\}$ \emph{covers} $p$ and $r$.  Similarly, a thin edge $(p, D)$ \emph{covers} $p$ and the resources in $D$.  A player or a resource is covered by a set of edges if it is covered by some edge in the set.  For any $w \geq 0$, a thin edge $(p, D)$ is \emph{$w$-minimal} if $\sum_{r \in D}v_r \geq w$ and $\sum_{r \in D'} v_r < w$ for any $D' \subsetneq D$.  For a $w$-minimal thin edge $(p, D)$,  it is not hard to see that $w \leq \sum_{r\in D}v_r < w + \lambda$.

Given the above definitions of fat and thin edges, finding the target allocation is equivalent to finding a set of mutually compatible edges that covers all the players.

\subsection{A local search idea}
The following local search idea is initially proposed by Asadpour et al.~\cite{AFS12}, and is also used in~\cite{AKS17, CM18a}.

Let $G$ be the bipartite graph formed by the players, the fat resources, and the fat edges.  We maintain a set $M$ of fat edges and a set $\cal E$ of thin edges such that: (i)~$M$ is a maximum matching of $G$, (ii)~edges in $\cal E$ are \emph{$\lambda$-minimal} and are mutually compatible, and (iii)~each player is covered by at most one edge in $M\cup {\cal E}$.  We call such $M$ and ${\cal E}$ a partial allocation.  Initially, $M$ is an arbitrary maximum matching of $G$, and $\cal E$ is empty.  The set $M\cup {\cal E}$ is updated and grown iteratively so that one more player is covered in each iteration.  The final set $M\cup {\cal E}$ covers all the players and induces our target allocation.

Let $p_0$ be a player not yet covered by $M \cup {\cal E}$.  We need to update $M \cup {\cal E}$ to cover $p_0$ without losing any player that are already covered.  The simplest case is that we can find a player $q_0$ such that $q_0$ is covered by a thin edge $a$ compatible with $\cal E$ and there is an alternating path~\cite{HK73} with respect to $M$ from $p_0$ to $q_0$.  Let $\pi$ be this alternating path.  We first update $M$ by taking the symmetric difference $M\oplus \pi$, i.e., remove the edges in $\pi \cap M$ from the matching and add the edges in $\pi \setminus M$ to the matching.  $M \oplus \pi$ is also a maximum matching of $G$. After the update, $p_0$ becomes matched while $q_0$ becomes unmatched.  Then we add $a$ to ${\cal E}$ to cover $q_0$ again.  Here we slight abuse the notion of alternating paths in the sense that wen allow an alternating path with no edge.  The $\oplus$ can easily extend to alternating paths with no edge.

It is possible that no edge covering $q_0$ is compatible with $\cal E$.  Let $a$ be an edge covering $q_0$.  Suppose that $b$ is an edge in $\cal E$ that is not compatible with $a$.  We say $b$ \emph{blocks} $a$.  Let $p_1$ be the player covered by $b$.   In order to add $a$ to ${\cal E}$, we have to release $b$ from $\cal E$.  But we cannot lose $p_1$, so before we release $b$, we need to find another edge to cover $p_1$.  Now $p_1$ has a similar role as $p_0$. 

\subsection{Node-disjoint alternating paths}
\label{sec:disjoint-path}

In order to achieve a polynomial running time, our algorithm updates $M$ using multiple node-disjoint alternating paths from unmatched players to players .  In this section, we define a problem of finding a largest set of node-disjoint paths.  We also extend the $\oplus$ operation to a set of node-disjoint paths.

For any maximum matching $M$ of $G$, we define $G_M$ to be the directed graph obtained from $G$ by orienting edges of $G$ from $r$ to $p$ if $\{p, r\} \in M$, and from $p$ to $r$ if $\{p, r\} \notin M$.  Let $S$ be a subset of the players not matched by $M$.  Let $T$ be a subset of the players.  Finding the largest set of node-disjoint alternating paths from $S$ to $T$ is equivalent to finding the largest set of node-disjoint paths in $G_M$ from $S$ to $T$.  Let $G_M(S,T)$ denote the problem of finding the largest set of node-disjoint paths from $S$ to $T$ in $G_M$.  Let $f_M(S, T)$ denotes the maximum number of such paths.  Note that when $S\cap T \neq \emptyset$, a path consisting of a single node is allowed.  Such path is called a trivial path.  Paths with at least one edge is non-trivial.  Let $\Pi$ be a feasible solution for $G_M(S,T)$. The paths in $\Pi$ originate from a subset of $S$, which we call the \emph{sources} and denote as $\src_{\Pi}$, and terminate in a subset of $T$, which we call the \emph{sinks} and denote as $\sink_{\Pi}$.  We extend the $\oplus$ operation to $\Pi$.  Viewing $\Pi$ as a set of edges, $M\oplus \Pi$ stands for removing the edges in $\Pi \cap M$ from the matching and adding the edges in $\Pi \setminus M$ to the matching.  One can see that $M\oplus \Pi$ is a maximum matching of $G$.

The problem $G_M(S, T)$ can be solved in polynomial time.  Please see the appendix~\ref{apd:path} for more about this problem.

\section{An Approximation Algorithm}
\label{sec:approx}

We discuss below a few techniques used by our algorithm. Some of them are used in~\cite{AKS17, CM18a, DRZ18}.  The \emph{limited blocking} strategy is brand new, and is crucial to achieving an approximation ratio of $4 +\delta$.  In the following discussion, one can interpret addable edges as thin edges that we hope to add to $\cal E$, and blocking edges as edges in $\cal E$ that are not compatible with addable edges.  The precise definition will be given later.

\emph{Layers}.  As in \cite{AKS17,CM18a}, we maintain a stack of layers, where each layer consists of addable edges and their blocking edges.  The key to achieving a polynomial running time is to guarantee a geometric growth in the number of blocking edges from the bottom to the top of the stack.

\emph{Layer-level node-disjoint paths}.  We require that the players covered by the addable edges in a layer can be simultaneously reached via node-disjoint paths in $G_M$ from the players covered by the blocking edges in the lower layers~\cite{DRZ18}.  It has the same effect as the globally node-disjoint path used in~\cite{AKS17}: if lots of addable edges in a layer become unblocked, then a significant update can be made.  The advantage of our strategy is that it offers more flexibility when building a new layer.

\emph{Lazy update}. When having an unblocked addable edge, one may be tempted to update $M$ and ${\cal E}$ immediately.  However, as in~\cite{AKS17}, in order to achieve a polynomial running time, we should wait until there are lots of unblocked addable edges, and then a significant update can be made in one step.  We will define a constant $\mu$ to control the laziness.

\emph{Greedy and Limited Blocking}.  Recall that the key to achieving a polynomial running time is to guarantee a geometric growth in the number of blocking edges from the bottom to the top of the stack.  In~\cite{AFS12},  every addable edge is $\lambda$-minimal, and each blocking edge blocks exactly one addable edge.  Using this strategy, in worst case, one may get a layer that has one addable edge being blocked by many blocking edges.  After some of these blocking edges are released from $\cal E$, we may be left with a layer that has a single addable edge being blocked by a single blocking edge, which breaks the geometric growth in the number of blocking edges.  To resolve this issue, Annamalai et al.~\cite{AKS17} allow a blocking edge to block as many addable edges as possible.  However, it causes a new trouble: one may get a layer that has many addable edges being blocked by one blocking edge.  Again, this breaks the geometric growth in the number of blocking edges.  As a consequence, they have to introduce another strategy \emph{Greedy}.  They require every addable edge to be $\frac{1}{2}$-minimal.  Such an addable edge contains much more resources than necessary.  If such an addable edge is blocked, at least $\frac{1}{2}-\lambda$ worth of its resources must be occupied by blocking edges.  Provided that a blocking edge is $\lambda$-minimal and covers at most $2\lambda$ worth of resources,  the greedy strategy ensures that, in a layer, the number of blocking edges cannot be too small comparing with the number of addable edges.  Analysis shows that although the greedy strategy makes the algorithm faster, it deteriorates the approximation ratio.  Our strategy is a generalization of those used in~\cite{AFS12} and~\cite{AKS17}.  We allow a blocking edge $b$ to block more than one addable edge, but once $b$ shares strictly more than $\beta\lambda$ worth of resources with the addable edges blocked by it,  we stop $e$ from blocking more edges.  We use greedy too.  In our algorithm, addable edges in layers are $(1+\gamma)\lambda$-minimal for some constant $\gamma$.  

If we set $\beta,\gamma,\mu$ to be $0$, then our algorithm acts exactly the same as the local search in~\cite{AFS12}, which achieves a ratio of $4$ but may not run in polynomial time.  If $\beta, \gamma$ are set to be some large constant, then our algorithm acts like the algorithm in~\cite{AKS17} which achieves a polynomial running time but a much worse ratio.  We show that carefully selected tiny $\beta$ and tiny $\gamma$ guarantee a polynomial running time but barely hurt the approximation ratio.

\subsection{The algorithm}
\label{sec:alg}

Let $M \cup {\cal E}$ be the current partial allocation.  Let $p_0$ be a player that is not yet covered by $M \cup {\cal E}$.  The algorithm alternates between two phases to update and extend $M \cup {\cal E}$ so that the partial allocation covers $p_0$ eventually without losing any covered player.  In the building phase, it pushes new layers onto a stack, where each layer stores some addable edges and their blocking edges.  In the collapse phase, it uses unblocked addable edges to release some blocking edges in some layer from $\cal E$.

Since we frequently talk about resources covered by thin edges and take sum of values over a set of resources, we define the following notations.  Given a thin edge $e$, $R_e$ denotes the set of resources covered by $e$.  Given a set $\cal S$ of thin edges, $R({\cal S})$ denotes the set of thin resources covered by $\cal S$.  Given a set $D$ of resources, define $v[D] = \sum_{r\in D} v_r$.

\subsubsection{Building phase}

The algorithm maintains a stack of layers.  The layer index starts with $1$ from the bottommost layer in the stack.  The $i$-th layer $L_i$ is a tuple $({\cal A}_i, {\cal B}_i, d_i, z_i)$, where ${\cal A}_i$ is a set of addable edges that we want to add to $\cal E$, and ${\cal B}_i$ is a set of blocking edges that prevent us from doing so. The two numeric values $d_i$ and $z_i$ are maintained for the sake of analysis.  The algorithm also maintains a set $\cal I$ of addable edges that are compatible with $\cal E$.  We will define addable edges and blocking edges later.  We use $\ell$ to denote the number of layers in the current stack.  The state of the algorithm is specified by $(M, {\cal E}, {\cal I}, (L_1, \ldots, L_{\ell}))$.

For each ${\cal A}_i$, we use $A_i$ to denote the set of players covered by ${\cal A}_i$. Similarly, $B_i$ and $I$ denote the set of players covered by ${\cal B}_i$ and $\cal I$, respectively.    For $i \in [1, \ell]$, define ${\cal B}_{\leq i} = \bigcup_{j = 1}^{i}{\cal B}_j$, $B_{\leq i} = \bigcup_{j=1}^i B_j$, and ${\cal A}_{\leq i} = \bigcup_{j=1}^i {\cal A}_j$.

For simplicity, we define the first layer $L_1$ to be $(\emptyset, \{(p_0, \emptyset)\}, 0, 0)$.   That is, ${\cal A}_1 = \emptyset$, ${\cal B}_1 = \{(p_0,\emptyset)\}$, and $d_1 = z_1 = 0$.  

The layers are built inductively.  Initially, there is only the layer $L_1$ and ${\cal I} = \emptyset$.  Let $\ell$ be the number of layers in the current stack.  Consider the construction of the $(\ell + 1)$-th layer.  

\begin{definition}
	\label{def:active-resrc}
	Let $\beta \geq 0$ be a constant to be specified later.  A thin resource $r$ is \textbf{inactive} if {\em (i)}~$r \in R({\cal A}_{\leq \ell}\cup {\cal B}_{\leq \ell})$, or {\em (ii)}~$r \in R({\cal A}_{\ell + 1} \cup {\cal I})$, or {\em (iii)}~$r \in R_b$ for some $b \in  {\cal B}_{\ell + 1}$ and $v[R_b \cap R({\cal A}_{\ell + 1})] >  \beta\lambda$.  If a thin resource is not inactive, then it is \textbf{active}.
\end{definition}

We will define addable edges so that they use only active thin resources.

\begin{definition}
	A player $p$ is \textbf{addable} if $f_M(B_{\leq \ell}, A_{\ell+1}\cup I \cup  \{p\}) = f_M(B_{\leq \ell}, A_{\ell+1}\cup I) + 1$.
\end{definition}

The activeness of thin resources and the addability of the players depend on ${\cal A}_{\ell+1}$ and ${\cal I}$ ($A_{\ell + 1}$ and $I$), so they may be affected as we add edges to ${\cal A}_{\ell+1}$ and ${\cal I}$.

\begin{definition}
\label{def:addable-edge}
A thin edge $(p, D)$ is \textbf{addable} if~$p$~is addable and $D$ is a set of \textbf{active} thin resources desired by $p$ with $v[D] \geq \lambda$.
The \textbf{blocking edges} of an addable edge $(p,D)$ are $\{\textrm{$e \in {\cal E}$ : $R_e \cap D\neq \emptyset$ }\}$.  An addable edge $(p, D)$ is \textbf{unblocked} if $v[D\setminus R(\cal E)] \geq \lambda$.  
\end{definition}

Recall that, for any $w > 0$, a thin edge $(p, D)$ is \emph{$w$-minimal} if $v[D] \geq w$ and $v[D'] < w$ for any $D' \subsetneq D$.  Our algorithm considers two kinds of addable edges.  The first kind is unblocked addable edges that are $\lambda$-minimal.  It is easy to see that if an unblocked addable edge is $\lambda$-minimal, then it must be compatible with ${\cal E}$.  We use $\cal I$ to keep such addable edges.  The second kind is blocked addable edges that are $(1+\gamma)\lambda$-minimal, where $\gamma$ is a constant to be specified later.  Such edges will be added to ${\cal A}_{\ell+1}$.  Once a $(1+\gamma)\lambda$-minimal addable edge $(p, D)$ becomes unblocked, we can easily extract a $\lambda$-minimal unblocked addable edge $(p, D')$ with $D' \subseteq D$.  

Consider condition (iii) in Definition~\ref{def:active-resrc}.  Let $b$ be a blocking edge in ${\cal B}_{\ell+ 1}$. When $R_b$ and $R({\cal A}_{\ell + 1})$ share strictly more than $\beta\lambda$ worth of resources, all resources in $R_b$ become inactive.  Any addable edge to be added to ${\cal A}_{\ell+1}$ in the future cannot use these inactive resources, and hence, will not be blocked by $b$. This is how we achieve ``limited blocking'' mentioned before.

We call {\sc Build} below to construct the $(\ell+1)$-th layer.  Note that after adding an addable edge to ${\cal A}_{\ell + 1}$, we immediately add its blocking edges to ${\cal B}_{\ell + 1}$ in order to keep the inactive/active status of thin resources up-to-date.

\begin{quote}
\noindent {\sc Build}$(M, {\cal E}, {\cal I}, (L_1,\cdots,L_\ell))$
\begin{enumerate}	
	\item Initialize ${\cal A}_{\ell+1} = \emptyset$ and ${\cal B}_{\ell+1} = \emptyset$.

	\item While there is an unblocked addable edge that is $\lambda$-minimal, add it to $\cal I$. \label{step:build-1}
	
	\item While there is an addable edge $(p, D)$ that is $(1+\gamma)\lambda$-minimal
	\begin{enumerate}[{3.}1]
		\item add $(p, D)$ to ${\cal A}_{\ell + 1}$. (Note that $(p, D)$ must be blocked; otherwise, we could extract from it a $\lambda$-minimal unblocked addable edge, which should be added to ${\cal I}$ in step~\ref{step:build-1}.)
		\item add to ${\cal B}_{\ell+1}$ the edges in $\cal E$ that block $(p, D)$.
	\end{enumerate}
	\item Set $d_{\ell+1}: = f_M(B_{\leq \ell}, A_{\ell+1}\cup I)$, $z_{\ell+1}: = \left|A_{\ell+1}\right|$, and $L_{\ell+1} := ({\cal A}_{\ell+1}, {\cal B}_{\ell + 1}, d_{\ell+1}, z_{\ell + 1})$.
	\item Update $\ell:= \ell + 1$
\end{enumerate}
\end{quote}

Table~\ref{tb:fact} lists a few facts about the layers.

\begin{table}
	\centering	
	\begin{tabu}to \textwidth {X[1cb]X[5lp]}
		\toprule
		Fact~1 & For every $i \in [1, \ell]$, edges in ${\cal A}_i$ are mutually compatible.\\
			&For every $i,j\in[1, \ell]$ with $i\neq j$, $R({\cal A}_i) \cap R({\cal A}_j) = \emptyset$.\\ 
		\midrule
		Fact~2 & Edges in ${\cal I}$ are mutually compatible, and they are also compatible with edges in $\cal E$. For every $i \in [1, \ell]$, $R({\cal A}_i) \cap R({\cal I}) = \emptyset$\\
		\midrule
		Fact~3 & $\{{\cal B}_2, \ldots, {\cal B}_\ell\}$ are disjoint subsets of ${\cal E}$. Note that ${\cal B}_1 = \{(p_0,\emptyset)\}$ does not share any resource with ${\cal B}_i$ for $i \in [2,\ell]$.  \\
		\bottomrule
	\end{tabu}
	\caption{
		Some facts about $\cal I$ and the layers in the stack.
	}
	\label{tb:fact}
\end{table}

\subsubsection{Collapse phase}
When some layer becomes collapsible, the algorithm enters the collapse phase.  Let $(M, {\cal E}, {\cal I}, \allowbreak (L_1,\cdots,L_\ell))$ be the current state of the algorithm. In order to determine whether a layer is collapsible or not, we need to compute the following decomposition of ${\cal I}$.  Let $({\cal I}_1, \ldots, {\cal I}_{\ell - 1})$ be some disjoint subsets of $\cal I$.  Let $I_i$ denote the set of players covered by ${\cal I}_i$.  For $i \in [1, \ell-1]$, we use ${\cal I}_{\leq i}$ and $I_{\leq i}$ to denote $\bigcup_{j=1}^i {\cal I}_j$ and $\bigcup_{j=1}^i I_j$, respectively. 

\begin{definition}
\label{def:canon}
A collection of disjoint subsets $({\cal I}_1, \ldots, {\cal I}_{\ell-1})$ of $\cal I$ is a \textbf{canonical decomposition} of $\cal I$ if for all $i \in [1,\ell-1]$, $f_M(B_{\leq i}, I_{\leq i}) = f_M(B_{\leq i}, I) = |I_{\leq i}|$.  A solution $\Gamma$ for $G_M(B_{\leq \ell-1},I)$ is a \textbf{canonical solution} with respect to the canonical decomposition $({\cal I}_1, \ldots, {\cal I}_{\ell-1})$ if $\Gamma$ can be partitioned into disjoint subsets $(\Gamma_1, \ldots, \Gamma_{\ell-1})$ such that for every $i \in [1,\ell-1]$, $\Gamma_i$ is a set of $|I_i|$ paths from $B_i$ to $I_i$ in $G_M$.
\end{definition}

Although it is not clear from the definition, invariant 1 in Table~\ref{tb:invar} implies that $({\cal I}_1, \ldots, {\cal I}_{\ell-1})$ is indeed a partition of $\cal I$.  The following lemma is analogous to its counterpart in~\cite{AKS17,CM18a}.  We put its proof to Appendix~\ref{apd:proofs}.

\begin{lemma}
\label{lem:canon}
Let $\ell$ be the number of layers in the stack.  A canonical decomposition of $\cal I$ and a corresponding canonical solution for $G_M(B_{\leq \ell-1 },I)$ can be computed in $\mathit{poly}(\ell,m,n)$ time. 
\end{lemma}

All the edges in ${\cal I}_i$ are compatible with ${\cal E}$, and all the players in $I_i$ can be reached from $B_i$ by node-disjoint paths $\Gamma_i$ in $G_M$, so every edge in ${\cal I}_i$ can be used to release one blocking edge in ${\cal B}_i$ from $\cal E$.  A layer is collapsible if a certain fraction of its blocking edges can be released.

\begin{definition}
Let $\mu$ be a constant to be specified later.  A layer $L_i$ is \textbf{collapsible} if there is a canonical decomposition $({\cal I}_1, \ldots, {\cal I}_{\ell-1})$ of $\cal I$ such that $|I_i| > \mu|B_i|$.
\end{definition}

Note that although there can be more than one canonical decomposition, the collapsibility of a layer is independent of the choice of canonical decompositions, because by definition, we always have $|I_i| = f_M(B_{\leq i}, I) - f_M(B_{\leq i-1}, I)$.  

When some layer is collapsible, we enter the collapse phase, call {\sc Collapse} to shrink collapsible layers until no layer is collapsible, and then return to the build phase.

\begin{quote}
\noindent {\sc Collapse}$(M, {\cal E}, {\cal I}, (L_1,\cdots,L_\ell))$
\begin{enumerate}[1.]
	\item Compute a canonical decomposition $({\cal I}_1, \ldots, {\cal I}_{\ell-1})$ and a corresponding canonical solution $\Gamma_1 \cup \cdots \cup \Gamma_{\ell-1}$ for $G_M(B_{\leq \ell-1},I)$.  If no layer is collapsible, go to the build phase; otherwise, let $L_t$ be the collapsible layer with the smallest index.\label{step:collapse-1}

	\item Remove all the layers above $L_t$ from the stack.  Set ${\cal I} := {\cal I}_{\leq t-1}$. \label{step:collapse-2}

	\item Recall that $\src_{\Gamma_t} \subseteq B_t$ by Definition~\ref{def:canon}.  Let ${\cal B}^{\Gamma}$ denote the set of edges in ${\cal B}_t$ that are incident to players in $\src_{\Gamma_t}$. We use ${\cal I}_t$ and $\Gamma_t$ to release the edges in ${\cal B}^{\Gamma}$.\label{step:collapse-3}
	\begin{enumerate}[{3.}1]
		\item Update $M$ by flipping the paths in $\Gamma_t$, i.e., set $M := M \oplus \Gamma_{t}$.
	
		\item Add to ${\cal E}$ the edges in ${\cal I}_t$, i.e., set ${\cal E} := {\cal E} \cup {\cal I}_t$.

		\item Each player in $\src_{\Gamma_t}$ is now covered by either a fat edge or a thin edge from ${\cal I}_t$. If $t = 1$, then $p_0$ is already covered, and the algorithm terminates.  Assume that $t \geq 2$.  Edges in ${\cal B}^{\Gamma}$ can be safely released from $\cal E$. Set ${\cal E} := {\cal E} \setminus {\cal B}^{\Gamma}$ and ${\cal B}_t := {\cal B}_t \setminus {\cal B}^{\Gamma}$. 
	\end{enumerate}

	\item If $t \geq 2$, we need to update ${\cal A}_t$ because some edges in ${\cal A}_t$ may become unblocked due to the release of blocking edges. For every edge $(p, D)$ in ${\cal A}_t$ that becomes unblocked, \label{step:collapse-4}
	\begin{enumerate}[{4.}1]
		\item Remove $(p,D)$ from ${\cal A}_t$,
		
		\item if $f_M(B_{\leq t-1}, I\cup \{p\}) = f_M(B_{\leq t-1}, I) + 1$, then extract a $\lambda$-minimal unblocked addable edge $(p,D')$ from $(p, D)$, and add $(p,D')$ to $\cal I$.
	\end{enumerate}  

	\item Update $\ell := t$.  Go to step~\ref{step:collapse-1}.
\end{enumerate}
\end{quote}
\section{Analysis of the approximation algorithm}
\label{sec:analysis}

\subsection{Some invariants}
Table~\ref{tb:invar} lists a few invariants, where $\ell$ is the number of layers in the stack.  Lemma~\ref{lem:invar} gives a few invariants maintained by the algorithm.  Lemma~\ref{lem:non-collapse} states that when no layer is collapsible, ${\cal I}$ cannot have too many unblocked addable edges, nor can a layer lose too many addable edges.  Lemmas~\ref{lem:invar} and~\ref{lem:non-collapse} below have analogous versions in~\cite{AKS17,CM18a} and can be proved similarly.  We put their proofs to Appendix~\ref{apd:proofs}.

\begin{lemma}
	\label{lem:invar}
	{\sc Build} and {\sc Collapse} maintain the invariants in Table~\ref{tb:invar}.
\end{lemma}

\begin{lemma}
	\label{lem:non-collapse}
	Let $(L_1,\ldots,L_\ell)$ be a stack of layers.  If no layer is collapsible, then 
	\begin{enumerate}[(i)]
	\item $|I| \leq \mu|B_{\leq \ell-1}|$, and 
	\item for all $i \in [1,\ell]$, $|A_i| \geq z_i - \mu |B_{\leq i-1}|$.
	\end{enumerate}
\end{lemma}

\begin{table}
	\begin{tabu} to \textwidth {X[1cb]X[5lp]}		
		\toprule
		Invariant~1 & $f_M(B_{\leq \ell-1},I) = |I|$. \\  
		\midrule
		Invariant~2 & For every $i \in [1,\ell-1]$, $f_M(B_{\leq i}, A_{i+1} \cup I) \geq d_{i+1}$. \\
		\midrule
		Invariant~3 & For every $i \in [1, \ell]$, $|A_i| \leq z_i$.\\
		\midrule
		Invariant~4 & For every $i \in [1, \ell]$, $d_i \geq z_i$. \\
		\bottomrule
	\end{tabu}
	\caption{
			Invariants maintained by the algorithm.
		}
	\label{tb:invar}
\end{table} 

\subsection{Bounding the number of blocking edges}
Lemma~\ref{lem:limit-overlap} --~\ref{lem:blocking-num-b} are consequences of our greedy and limited blocking strategies.  Basically, they bound the number of blocking edges in a layer in terms of the number of addable edges.

\begin{lemma}
\label{lem:limit-overlap}
Let $L_{i} = ({\cal A}_i, {\cal B}_i,  d_i, z_i)$ be an arbitrary layer in the stack.  For each edge $b \in {\cal B}_i$, there is an edge $a\in {\cal A}_i$ such that $v[R_b \cap R({\cal A}_i \setminus \{a\})] \leq \beta\lambda$.
\end{lemma}
\begin{proof}
	Let $b$ be an edge in ${\cal B}_i$.  Sort the edges in ${\cal A}_i$ in chronological order of their additions into ${\cal A}_i$.  Let $a$ be the last edge in ${\cal A}_i$ that is blocked by $b$.  By our choice of $a$, the edges in ${\cal A}_i$ that are added after $a$ cannot be blocked by $b$, so they do not share any common resource with $b$.  Among the edges added before $a$, let $S$ be the subset of their resources that are also covered by $b$.  We claim that $v[S] \leq \beta\lambda$.  If not, all resources in $R_b$ would be inactive before the addition of $a$ by definition of inactive resources.  So no resource in $R_b$ could be included in $a$.  But $R_a \cap R_b$ must be non-empty as $b$ blocks $a$, a contradiction.
\end{proof}

\begin{lemma}
\label{lem:blocking-num-a}
	Let $L_{i} = ({\cal A}_i, {\cal B}_i,  d_i, z_i)$ be an arbitrary layer in the stack. We have $|A_i| < (1 + \frac{\beta}{\gamma})|B_i|$.
\end{lemma}
\begin{proof}
	By Lemma~\ref{lem:limit-overlap}, for each $b \in {\cal B}_i$, we can identify an edge $a_b \in {\cal A}_i$ so that $v[R_b \cap R({\cal A}_i \setminus \{a_b \})] \leq \beta\lambda$.
	Let ${\cal A}^0_i = \{a_b : b \in {\cal B}_i\}$ be the set of edges identified.  $|{\cal A}^0_i| \leq |B_i|$.  

	Let ${\cal A}^1_i = {\cal A}_i \setminus {\cal A}^0_i$.  For every $b \in {\cal B}_i$, $v[R_b\cap R({\cal A}_i^1)] \leq v[R_b \cap R({\cal A}_i \setminus \{a_b\})] \leq \beta\lambda$.  Taking sum over all edges in ${\cal B}_i$, we get $v[R({\cal B}_i)\cap R({\cal A}_i^1)] \leq \beta\lambda|B_i|$.  On the other hand, each edge $a$ in ${\cal A}_i^1$ is $(1 + \gamma)\lambda$-minimal and is blocked, so it must have more than $\gamma\lambda$ worth of its resources occupied by edges in ${\cal B}_i$, i.e., $v[R({\cal B}_i) \cap R_a] > \gamma\lambda$.  Summing over all the edges in ${\cal A}^1_i$ gives $v[R({\cal B}_i)\cap R({\cal A}_i^1)] > \gamma\lambda|{\cal A}^1_i|$.  Hence, $\beta\lambda|B_i| \geq v[R({\cal B}_i)\cap R({\cal A}_i^1)] > \gamma\lambda|{\cal A}^1_i|$.  

	Finally we get $|A_i| = |{\cal A}^0_i| + |{\cal A}^1_i| < |B_i| + \frac{\beta}{\gamma}|B_i| = (1 + \frac{\beta}{\gamma})|B_i|$.
\end{proof}

\begin{lemma}
\label{lem:blocking-num-b}
	Let $L_{i} = ({\cal A}_i, {\cal B}_i,  d_i, z_i)$ be an arbitrary layer in the stack. Let ${\cal B}'_i$ be the set of edges in ${\cal B}_i$ that share strictly more than $\beta\lambda$ resources with edges in ${\cal A}_i$.  More precisely, ${\cal B}'_i = \{e \in {\cal B}_{i} : v[R_e \cap R({\cal A}_{i})] > \beta\lambda\}$.
	We have $|{\cal B}'_i| < \frac{2+\gamma}{\beta}|A_{i}|$.
\end{lemma}
\begin{proof}
	Taking the sum of $v[R_e \cap R({\cal A}_{i})]$ over all edges $e$ in ${\cal B}'_i$, we obtain 
		$v[R({\cal B}'_i)\cap R({\cal A}_{i})] > \beta\lambda|{\cal B'}_i|$.
	On the other hand, 
		$v[R({\cal B}'_i)\cap R({\cal A}_{i})] \leq v[R({\cal A}_{i})] < (2+\gamma)\lambda|A_{i}|$.
	The last inequality is because that edges in ${\cal A}_i$ are $(1 + \gamma)\lambda$-minimal.
	Combining the above two inequalities, we obtain
		$\beta\lambda|{\cal B'}_i| < (2+\gamma)\lambda|A_{i}| \Rightarrow |{\cal B}'_i| < \frac{2+\gamma}{\beta}|A_{i}|$.
\end{proof}

\subsection{Geometric growth in the number of blocking edges}
Now we are ready to prove that the number of blocking edges grow geometrically from bottom to top.  Lemma~\ref{lem:key1} states that there are lots of addable edges in a layer \emph{immediately after} its construction.  Previous Lemma~\ref{lem:non-collapse}(ii) ensures that as long as there is no collapsible layer, every layer cannot lose too many addable edges.  Therefore, Lemma~\ref{lem:key1} and Lemma~\ref{lem:non-collapse}(ii) together imply that as long as no layer is collapsible, every layer has lots of addable edges.  Since the number of blocking edges in a layer is lower bounded in terms of the number of addable edges by Lemma~\ref{lem:blocking-num-a}, we can conclude that there must be lots of blocking edges in a layer when no layer in the stack is collapsible (Lemma~\ref{lem:key}).

\begin{lemma}
	\label{lem:key1}
	Let $(M, {\cal E},{\cal I}, (L_1, \ldots, L_{\ell+1}))$ be the state of the algorithm \textbf{immediately after} the construction of $L_{\ell+1}$. If no layer is collapsible, then $z_{\ell+ 1} = |{\cal A}_{\ell+1}| \geq 2\mu |{\cal B}_{\leq \ell}|$.
\end{lemma}
\begin{proof}
	We give a proof by contradiction.  Suppose that $z_{\ell+ 1} = |{\cal A}_{\ell+1}| < 2\mu |{\cal B}_{\leq \ell}|$.  We will show that the dual of $\mathit{CLP}(1)$ is unbounded, which implies that $\mathit{CLP}(1)$ is infeasible, contradicting the assumption that the configuration LP has optimal value $T^* = 1$.

	Consider the moment immediately after we finish adding edges to ${\cal A}_{\ell + 1}$ during the construction of $L_{\ell+1}$.  At this moment, there is no $(1 + \gamma)\lambda$-minimal addable edge left.  The rest of the proof is with respect to this moment.  

	Let $\Pi$ be an optimal solution for $G_M(B_{\leq \ell}, A_{\ell+1} \cup I)$.  Note that $M \oplus \Pi$ is a maximum matching of $G$.  Recall that $G_{M \oplus \Pi}$ is a directed graph  obtained from $G$ by orienting the edges of $G$ according to whether they are in $M \oplus \Pi$ or not.  Let $P^+$ be the set of players that can be reached in $G_{M\oplus \Pi}$ from $B_{\leq \ell}\setminus \src_\Pi$.  Let $R^+_f$ be the set of fat resources that can be reached in $G_{M\oplus \Pi}$ from $B_{\leq \ell}\setminus \src_\Pi$.   Let $R^+_t$ be the set of inactive thin resources.

	\begin{claim}
		\label{cl:residual-player}
		(i)~Players in $P^+$ are still addable after we finish adding edges to ${\cal A}_{\ell + 1}$  (ii)~Players in $P^+$ have in-degree at most 1 in $G_{M \oplus \Pi}$. (iii)~Resources in $R^+_f$ have out-degree exactly $1$ in $G_{M \oplus \Pi}$.
	\end{claim}

	We define a dual solution $(\{y^*_p\}_{p\in P}, \{z^*_r\}_{r\in R})$ as follows.  
	\begin{alignat*}{4}
			y^*_p &= \left\{
				\def\arraystretch{1.5}
				\begin{array}{ll}
					1 - (1+\gamma)\lambda &\text{if $p \in P^+$,}\\
					0 &\text{otherwise.}
				\end{array}
		\right. & \qquad z^*_r &= \left\{
				\def\arraystretch{1.5}
				\begin{array}{ll}
					1 - (1 + \gamma)\lambda 
						&\text{if $r\in R^+_f$,}\\
					v_r 
						&\text{if $r \in R^+_t$,}\\
					0
						&\text{otherwise.}
				\end{array}
				\right.
	\end{alignat*}

	\begin{claim}
		\label{cl:approx-dual}
		$(\{y^*_p\}_{p\in P}, \{z^*_r\}_{r\in R})$ is a feasible solution, and it has a positive objective function value.
	\end{claim}

	Suppose that Claim~\ref{cl:approx-dual} holds. Then $(\{\alpha y^*_p\}_{p\in P}, \{\alpha z^*_r\}_{r\in R})$ is also a feasible solution for any $\alpha > 0$.  As $\alpha$ goes to infinity, the objective function value goes to infinity, yielding the contradiction that we look for.
\end{proof}

We defer the proof of Claim~\ref{cl:residual-player} to Appendix~\ref{apd:proofs}, and we give the proof of Claim~\ref{cl:approx-dual} below.

\begin{claimproof}[Feasibility] We need to show that $\forall\, p\in P, \,\, \forall \, C \in {\cal C}_p(1)$, $y^*_p \leq \sum_{r \in C}z^*_r$.	
	If $p \notin P^+$, then $y^*_{p} = 0$, and the inequality holds since $z^*_r$ is non-negative.  Assume that $p \in P^+$.  So $y^*_p = 1 - (1+\gamma)\lambda$.  Let $C$ be any configuration for $p$.  We show that $\sum_{r \in C}z^*_r \geq 1 - (1+\gamma)\lambda$.
	
	Case~1: $C$ contains a fat resource $r_f$.  Since $p$ desires $r_f$, $G_{M\oplus \Pi}$ has either an edge $(p,r_f)$ or an edge $(r_f,p)$. By the definition of $P^+$, there is a path $\pi$ in $G_{M \oplus \Pi}$ from $B_{\leq \ell}\setminus \src_\Pi$ to $p$.  If $G_{M\oplus \Pi}$ has an edge $(p,r_f)$, we can reach $r_f$ from $B_{\leq \ell}\setminus \src_\Pi$ by following $\pi$ and then $(p,r_f)$.  So $r_f \in R^+_f$.  If $G_{M\oplus \Pi}$ has an edge $(r_f,p)$, then $p$ is matched with $r_f$ by $M\oplus \Pi$.  Since players in $B_{\leq \ell}\setminus \src_\Pi$ are not matched by $M\oplus \Pi$, $p \notin (B_{\leq \ell}\setminus \src_\Pi)$.  By Claim~\ref{cl:residual-player}, in $G_{M \oplus \Pi}$, the in-degree of $p$ is at most one, so $(r_f,p)$ is the only edge entering $p$. To reach $p$, $\pi$ must reach $r_f$ first.  Hence, we can follow $\pi$ to reach $r_f$ from $B_{\leq \ell}\setminus \src_\Pi$, which implies $r_f \in R^+_f$.  In both cases, we have $\sum_{r \in C}z^*_r  \geq z^*_{r_f} = 1 - (1 + \gamma)\lambda$.

	Case~2: $C$ contains only thin resources.  By Claim~\ref{cl:residual-player}, $p$ is still addable after we finish adding edges to ${\cal A}_{\ell + 1}$.  However, when we finish adding edges to ${\cal A}_{\ell+1}$, there is no $(1+\gamma)\lambda$-minimal addable edge left.  It must be that $p$ does not have enough active resources to form an addable edge.  The active thin resources in $C$ must have a total value less than $(1+ \gamma)\lambda$.  Recall that $v[C] \geq 1$. At least $1 - (1 + \gamma)\lambda$ worth of thin resources in $C$ are inactive.  Since $z^*_r = v_r$ for inactive thin resources, $\sum_{r\in C}z^*_r \geq 1 - (1 + \gamma)\lambda$.
\end{claimproof}

\begin{claimproof}[Positive Objective Function Value]
	We need to show that $\sum_{p\in P}y^*_p - \sum_{r\in R}z^*_r > 0$.
	By our setting of $y^*_p$ and $z_r^*$, $\sum_{p \in P}y^*_p - \sum_{r\in R}z^*_r = \sum_{p \in P^+}y^*_p - \sum_{r\in R^+_f}z^*_r - \sum_{r\in R^+_t}z^*_r$.

	First consider $\sum_{p \in P^+}y^*_p - \sum_{r\in R^+_f}z^*_r$.  Since $y^*_p$ and $z^*_r$ have the same value $1 - (1+\gamma)\lambda$ for $p \in P^+$ and $r \in R^+_f$, it suffices to bound $|P^+| - |R^+_f|$ from below.  For each $r_f \in R^+_f$, by Claim~\ref{cl:residual-player}, $r_f$ has exactly one out-going edge to some player $p$ in $G_{M \oplus \Pi}$.   Since $r_f$ is reachable from $B_{\leq \ell}\setminus \src_\Pi$, so is $p$.  That is, $p \in P^+$.  We charge $r_f$ to $p$.  By Claim~\ref{cl:residual-player}, each player in $P^+$ has in-degree at most $1$ in $G_{M \oplus \Pi}$, so each of them is charged at most once.  Note that players in $B_{\leq \ell}\setminus \src_\Pi$ obviously belong to $P^+$ because they can be reached by themselves.  Moreover, they have zero in-degree in $G_{M\oplus \Pi}$ as they are not matched by $M\oplus \Pi$, so they are not charged.  Therefore, $|P^+| - |R^+_f| \geq |B_{\leq \ell}\setminus \src_\Pi| \geq |B_{\leq \ell}| - |A_{\ell + 1}| - |I|$. The last inequality is because $\Pi$ is an optimal solution for $G_M(B_{\leq \ell}, A_{\ell+1}\cup I)$.  In summary, 
	\begin{equation}
		\label{eq:lower}
		\sum_{p \in P^+}y^*_p - \sum_{r\in R^+_f}z^*_r \geq \left(1 - (1 + \gamma)\lambda\right)\left(|B_{\leq \ell}| - |A_{\ell + 1}| - |I|\right).
	\end{equation}

	Now consider $\sum_{r\in R^+_t}z^*_r$. By definition of inactive resources, $R^+_t$ can be divided into three parts: those covered by ${\cal A}_{\leq \ell} \cup {\cal B}_{\leq \ell}$, those covered by ${\cal A}_{\ell+1}\cup {\cal I}$, and those covered by ${\cal B}'_{\ell+1} = \{e \in {\cal B}_{\ell+1} : v[R_e \cap R({\cal A}_{\ell + 1})] > b\lambda\}$.  We handle these three parts separately.  

	Every edge in ${\cal A}_{\leq \ell}$ is blocked by some edges in ${\cal B}_{\leq \ell}$, so it has less than $\lambda$ worth of thin resources not used by ${\cal B}_{\leq \ell}$.  Every edge in ${\cal B}_{\leq \ell}$ is $\lambda$-minimal, so it covers less than $2\lambda$ worth of thin resources.  Thus, 
		$v[R({\cal A}_{\leq \ell} \cup {\cal B}_{\leq \ell})] < \lambda|{\cal A}_{\leq \ell}| + 2\lambda|{\cal B}_{\leq \ell}| 
				\leq \left(3 + \frac{\beta}{\gamma}\right)\lambda|B_{\leq \ell}|$.
	The last inequality is by Lemma~\ref{lem:blocking-num-a}.  Edges in ${\cal A}_{\ell+1}$ are $(1+\gamma)\lambda$-minimal, so each of them covers less than $(2 + \gamma)\lambda$ worth of resources. Edges in ${\cal I}$ are $\lambda$-minimal, so each of them covers less than $2\lambda$ worth of thin resources.  Therefore, 
	$v[R({\cal A}_{\ell+1}\cup {\cal I})] < (2 + \gamma)\lambda|A_{\ell + 1}| + 2\lambda|I|$.
	Edges in ${\cal B}'_{\ell+1}$ are $\lambda$-minimal, so 
		$v[R({\cal B}'_{\ell+1})] < 2\lambda|{\cal B}'_{\ell+1}| < \frac{4 + 2\gamma}{\beta}\lambda|A_{\ell +1}|$.
	The second inequality is by Lemma~\ref{lem:blocking-num-b}.
	Combining the above three parts gives
	\begin{equation}
		\label{eq:ulti-upper}
		\sum_{r\in R^+_t}z^*_r = \sum_{r\in R^+_t}v_r < \left(3 + \frac{\beta}{\gamma}\right)\lambda|B_{\leq \ell}| + 2\lambda|I| + \left(2 + \gamma + \frac{4+ 2\gamma}{\beta}\right)\lambda|A_{\ell + 1}|.
	\end{equation}

	Combining \eqref{eq:lower} and \eqref{eq:ulti-upper} gives that 
	\begin{align*}
		&\sum_{p \in P^+}y^*_p - \sum_{r\in R^+_f}z^*_r - \sum_{r\in R^+_t}z^*_r\\
	> &	\left(1 - \left(4 + \gamma + \frac{\beta}{\gamma}\right)\lambda\right)|B_{\leq \ell}| - \left(1 + \left(1 + \frac{4+2\gamma}{\beta}\right)\lambda\right)|A_{\ell + 1}| 
		- \left(1 + (1-\gamma)\lambda\right)|I|.
	\end{align*}

	By the contrapositive assumption at the beginning of the proof of Lemma~\ref{lem:key1}, $|A_{\ell + 1}| < 2\mu|B_{\leq\ell}|$.  Moreover, since no layer is collapsible, by Lemma~\ref{lem:non-collapse}(i), $|I| \leq \mu|B_{\leq \ell}|$.  Substituting these two inequalities into the above gives
	\begin{align*}
		&\sum_{p \in P^+}y^*_p - \sum_{r\in R^+_f}z^*_r - \sum_{r\in R^+_t}z^*_r\\
	> &	\left((1 - 3\mu) - \left(4 + \gamma + \frac{\beta}{\gamma} + 3\mu + \frac{(8 + 4\gamma)\mu}{\beta} -\gamma\mu\right)\lambda \right)|B_{\leq \ell}|.
	\end{align*}
	Let $\beta = \gamma^2$ and $\mu = \gamma^3$. We have 
	\[
		\sum_{p \in P^+}y^*_p - \sum_{r\in R^+_f}z^*_r - \sum_{r\in R^+_t}z^*_r
	> 	\left(\left(1 - 3\gamma^3\right) - \left(4 + 10\gamma + 4\gamma^2 +  3\gamma^3 - \gamma^4\right)\lambda \right)|B_{\leq \ell}|.
	\]
	Recall that $\lambda = \frac{1}{4+\delta}$ for some $\delta > 0$.  As $\gamma \to 0$, $\frac{1 - 3\gamma^3}{4 + 10\gamma + 4\gamma^2 +  3\gamma^3 - \gamma^4} \to \frac{1}{4}$.
	Hence, there is a sufficiently small $\gamma$ that makes $\frac{1 - 3\gamma^3}{4 + 10\gamma + 4\gamma^2 +  3\gamma^3 - \gamma^4} > \frac{1}{4+\delta} = \lambda$, thereby proving \[
	\sum_{p \in P}y^*_p -\sum_{r \in R}z^*_r = \sum_{p \in P^+}y^*_p - \sum_{r\in R^+_f}z^*_r - \sum_{r\in R^+_t}z^*_r > 0.\]  Moreover, one can verify that $\frac{1}{\gamma} = O(\frac{1}{\delta})$.
\end{claimproof}

\begin{lemma}
	\label{lem:key}
	Let $(M, {\cal E},{\cal I}, (L_1, \ldots, L_\ell))$ be a state of the algorithm. If no layer  is collapsible, then for $i\in [1, \ell-1]$, $|B_{i+1}| \geq \frac{\gamma^3}{1 + \gamma}|B_{\leq i}|$.
\end{lemma}
\begin{proof}
	Fix an $i \in [1,\ell-1]$.  Consider the period from the \emph{most recent} construction of layer $L_{i+1}$ until now.  During this period, none of the layers below $L_{i+1}$ has ever been collapsed; otherwise, $L_{i+1}$ would be removed, contradiction.  Hence, blocking edges in the layers below $L_{i+1}$ have never been touched during this period.  In other words, at the time $L_{i+1}$ was constructed, the set of blocking edges in the layers below $L_{i+1}$ was exactly ${\cal B}_{\leq i}$.  Also the constant $z_{i+1}$ is unchanged. By Lemma~\ref{lem:key1}, $z_{i+1} \geq  2\mu |{\cal B}_{\leq i}|$.  Although addable edges may be removed from $L_{i+1}$ during this period, there are still lots of addable edges left. 
	By Lemma~\ref{lem:non-collapse}(ii), $|{A}_{i+1}| \geq z_{i+1} - \mu |{B}_{\leq i}| \geq \mu |{B}_{\leq i}|$.
	By Lemma~\ref{lem:blocking-num-a}, $|B_{i+1}| > \frac{1}{(1 + \beta/\gamma)}|A_{i+1}| \geq \frac{\mu}{(1 + \beta/\gamma)}|{\cal B}_{\leq i}|$.
	Recall that we set $\beta = \gamma^2$ and $\mu = \gamma^3$ in the proof of Claim~\ref{cl:approx-dual}.  Replacing $\beta$ by $\gamma^2$ and $\mu$ by $\gamma^3$ proves the lemma.
\end{proof}

\begin{lemma}
	\label{lem:run-time}
	In $\mathit{poly}(m,n)\cdot n^{\mathit{poly}(\frac{1}{\delta})}$ time, the algorithm extends $M\cup {\cal E}$ to cover one more player.
\end{lemma}
Given Lemma~\ref{lem:key}, Lemma~\ref{lem:run-time} can be proved in a way similar to that of~\cite{AKS17, CM18a}. We sketch the proof here. Consider all non-collapsible states ever reached by the algorithm. By non-collapsible, we mean that no layer is collapsible in this state. Let $h = \frac{\gamma^3}{1+\gamma}$.  For each non-collapsible state $(M, {\cal E},{\cal I}, (L_1, \ldots, L_\ell))$, we define its signature vector $(s_1, \ldots, s_{\ell}, \infty)$ where $s_i = \log_{1/(1-\mu)}\frac{|B_i|}{h^{i+1}}$.  One can verify that the coordinates of the signature vector are non-decreasing, and that as the algorithm goes from one non-collapsible state to another, the signature vector decreases lexicographically. Moreover, the sum of the coordinates is bounded by $U^2$ where $U = \log n \cdot O(\frac{1}{\mu h}\log\frac{1}{h})$.  Each signature can be regarded as a partition of an integer less than or equal to $U^2$.  Summing up the number of partitions of an integer $i$ over all $i \in [1, U^2]$, we get the upper bound of $n^{O(\frac{1}{uh}\log\frac{1}{h})}$ on the number of distinct signatures.  Recall that $u = \gamma^3$, $h = \frac{\gamma^3}{1 + \gamma}$, and $\frac{1}{\gamma} = O(\frac{1}{\delta})$.  As a consequence, the number of non-collapsible states ever reached by the algorithm is bounded by $n^{\mathit{poly}(\frac{1}{\delta})}$.  Between two consecutive non-collapsible states, there is one {\sc Build} and at most $\log_{h+1}n$ {\sc Collapse}, which take $\mathit{poly}(m,n)\cdot n^{\mathit{poly}(\frac{1}{\delta})}$ time in total.  The total running time is thus $\mathit{poly}(m,n)\cdot n^{\mathit{poly}(\frac{1}{\delta})}$.



\bibliography{falloc}

\newpage
\appendix

\section{Node-disjoint Alternating Paths}
\label{apd:path}
Let $M$ be a maximum matching of $G$.  Let $S$ be a subset of the players not matched by $M$.  Let $T$ be a subset of the players.  Recall that the problem $G_M(S, T)$ seeks the largest set of node-disjoint paths in $G_M$ from $S$ to $T$, and $f_M(S, T)$ denotes the maximum number of such paths.  One may already observe that this problem can be easily solved after being reduced to a maximum flow problem.  However, for the sake of future analysis, we understand it from the perspective of matchings.  The following lemma is analogous to the well-known sufficient and necessary condition for maximum flow.

\begin{lemma}
	\label{lem:max-paths}
	Let $M$ be a maximum matching of $G$.  Let $S$ be a subset of the players not matched by $M$.  Let $T$ be a subset of $P$.  Let $\Pi$ be a feasible solution for $G_M(S,T)$.  $\Pi$ is an optimal solution for $G_M(S,T)$ if and only if there is no path in $G_{M \oplus \Pi}$ from $S\setminus \src_{\Pi}$ to $T\setminus \sink_\Pi$.  (Recall that $M\oplus \Pi$ is a maximum matching, so $G_{M \oplus \Pi}$ is well-defined.)
\end{lemma}
\begin{proof}[Proof of if part]
	Suppose that $\Pi$ is not optimal.  Recall that every path in $\Pi$ is an alternating path with respect to $M$.  We show that, with respect to the maximum matching $M \oplus \Pi$, there is an alternating path from $S\setminus \src_{\Pi}$ to $T\setminus \sink_\Pi$. This alternating path is a path in $G_{M \oplus \Pi}$ from  $S\setminus \src_{\Pi}$ to $T\setminus \sink_\Pi$.  First we observe that $\Pi$ cannot use any player in $S\setminus \src_{\Pi}$ because every of these players is not matched by $M$ and has in-degree $0$ in $G_M$.

	If there is a player $p \in (S\setminus \src_{\Pi})$ that belongs to $T$, we claim that $p$ itself forms our target alternating path.  $p$ is not matched by $M$ and is not used by alternating paths in $\Pi$, so $p$ remains unmatched in $M \oplus \Pi$ and $p \notin \sink_{\Pi}$. The player $p$ itself forms a trivial alternating path with respect to $M\oplus \Pi$, and it is from $S\setminus \src_{\Pi}$ to $T\setminus \sink_{\Pi}$.

	Assume that no player in $S\setminus \src_{\Pi}$ belongs to $T$.  Let $\Pi^*$ be an optimal solution for $G_M(S,T)$.  Clearly, $|\Pi^*| > |\Pi|$.  Consider $\src_{\Pi^*}\setminus \src_{\Pi}$.  It is non-empty as $|\Pi^*| > |\Pi|$.  No node in $\src_{\Pi^*}\setminus \src_{\Pi}$ belongs to $T$ as $\src_{\Pi^*}\setminus \src_{\Pi}$ is a subset of $S\setminus \src_{\Pi}$.  The alternating paths in $\Pi^*$ that start with nodes in $\src_{\Pi^*}\setminus \src_{\Pi}$ must be non-trivial as these nodes are not in $T$.
  	Now consider the edge set 
  	\[
  		E_{\oplus} = (M \oplus \Pi) \oplus \left(M \oplus \Pi^*\right) = \Pi \oplus \Pi^*.
  	\]
  	$E_{\oplus}$ is the symmetric difference of two maximum matchings, so it consists of some non-trivial even-length paths and possibly some cycles~\cite{HK73}.  All the paths in $E_{\oplus}$ are alternating paths with respect to $M\oplus \Pi$, and they are node-disjoint.  We claim that 
  	\begin{enumerate}[(i)]
  		\item every player in $\src_{\Pi^*} \setminus \src_{\Pi}$ is an endpoint of some non-trivial path in $E_{\oplus}$, and 
  		\item if a non-trivial path has an endpoint in $\src_{\Pi^*} \setminus \src_{\Pi}$, then the other endpoint must be in either $\src_{\Pi} \setminus \src_{\Pi^*}$ or $\sink_{\Pi^*} \setminus \sink_{\Pi}$.
  	\end{enumerate}
	Suppose that our claim holds.  Because $|\Pi^*| > |\Pi|$, we have $|\src_{\Pi^*} \setminus \src_{\Pi}| > |\src_{\Pi} \setminus \src_{\Pi^*}|$.  By pigeonhole principle, there must be at least one path in $E_\oplus$ that goes from $\src_{\Pi^*} \setminus \src_{\Pi}$ to $\sink_{\Pi^*} \setminus \sink_{\Pi}$, and this is our target alternating path with respect to $M\oplus \Pi$.

	To see why our claim holds, first consider the degrees of nodes in $E_{\oplus}$.  Viewing $E_{\oplus}$ as $\Pi \oplus \Pi^*$, we observe that 
	\begin{itemize}
		\item every player in $\src_{\Pi^*} \setminus \src_{\Pi}$ has degree $1$ in $E_{\oplus}$, because they have degree $1$ in $\Pi^*$ and degree $0$ in $\Pi$;
		\item every player in $\src_{\Pi} \setminus \src_{\Pi^*}$, $\sink_{\Pi^*} \setminus \sink_{\Pi}$, and $\sink_{\Pi} \setminus \sink_{\Pi^*}$ may have odd degree in $E_{\oplus}$;
		\item any other players or fat resources must have even degree in $E_{\oplus}$, and hence they cannot be endpoints of paths in $E_{\oplus}$.
	\end{itemize}
	Our claim (i) simply follows by observation (i).  To prove claim (ii), it suffices to show that if a path has an endpoint in $\src_{\Pi^*} \setminus \src_{\Pi}$, then its other endpoint cannot be in $\src_{\Pi^*} \setminus \src_{\Pi}$ or $\sink_{\Pi} \setminus \sink_{\Pi^*}$.  Recall that all the paths in $E_{\oplus}$ are alternating paths with respect to $M\oplus \Pi$.  Players in $(\src_{\Pi^*} \setminus \src_{\Pi}) \subseteq (S \setminus \src_{\Pi})$ are not matched by $M\oplus \Pi$, so they cannot be reached by a non-trivial alternating path with respect to $M\oplus \Pi$ from another unmatched node in $\src_{\Pi^*} \setminus \src_{\Pi}$.  Similarly, players in $\sink_{\Pi} \setminus \sink_{\Pi^*}$ are not matched by $M\oplus \Pi$, so cannot be reached by a non-trivial alternating path from another unmatched node in $\src_{\Pi^*} \setminus \src_{\Pi}$.  This completes the proof.
\end{proof}

\begin{proof}[Proof of only-if part] Let $\pi$ be a path in $G_{M\oplus \Pi}$ from $S\setminus \src_{\Pi}$ to $T\setminus \sink_\Pi$.  In other words, $\pi$ is an alternating path with respect to $M\oplus \Pi$ from $S\setminus \src_{\Pi}$ to $T\setminus \sink_\Pi$.  We show how to update $\Pi$ to a large feasible solution for $G_M(S, T)$.  If $\pi$ is a trivial alternating path consisting of a player $p$ alone, then $p$ must be in $S\cap T$, and hence is not matched by $M$. Also $p$ must be node-disjoint from $\Pi$.  Therefore, $p$ itself forms a (trivial) alternating path with respect to $M$, and adding this path to $\Pi$ yields a larger feasible solution for $G_M(S, T)$.

Suppose that $\pi$ is non-trivial.  Let $\src_{\pi}$ be the starting point of $\pi$, and let $\sink_{\pi}$ be the ending point of $\pi$.  Let $\Pi^0$ be the set of trivial alternating paths in $\Pi$, and let $\Pi^+$ be the set of non-trivial alternating paths in $\Pi$.  Note that players involved in $\Pi^0$ remain unmatched in $M\oplus \Pi$ and have in-degree 0 in $G_{M\oplus \Pi}$, so they can not be touched by the non-trivial path $\pi$ which starts at another player.  $\Pi^0$ and $\pi$ are node-disjoint.  Consider the edge set 
\[
	E_{\oplus} = M \oplus (M \oplus \Pi \oplus \pi) = \Pi \oplus \pi = \Pi^+ \oplus \pi.
\]
$E_{\oplus}$ is the symmetric difference of two maximum matchings, so it consists of some non-trivial even-length paths and possibly some cycles.  All the paths in $E_{\oplus}$ are alternating paths with respect to $M$, and they are node-disjoint.  There are two cases: (i) $\sink_{\pi} \notin \src_{\Pi^+}$, and (ii)$\sink_{\pi} \in \src_{\Pi^+}$.
In case (i), similar to that in the proof of if part, if viewing $E_{\oplus}$ as the symmetric difference of two sets of paths, players in $\src_{\Pi^+} \cup \{\src_{\pi}\}$ and $\sink_{\Pi}\cup \{\sink_{\pi}\}$ have degree $1$ in $E_{\oplus}$, so they are endpoints of paths in $E_{\oplus}$.  All the other players and fat resources have even degree.  As a result, there are exactly $|\src_{\Pi^+}| + 1 = |\Pi^+| + 1$ non-trivial paths in $E_{\oplus}$.  These paths, together with trivial paths in $\Pi^0$, form a larger feasible solution for $G_M(S, T)$.  In case (ii), one player $p^* = \sink_{\pi}\in \src_{\Pi^+}$ has degree $0$ or $2$ in $E_{\oplus}$, so $p^*$ cannot be an endpoint of any path in $E_{\oplus}$.  As a consequence, $E_{\oplus}$ contains only $|\Pi^+|$ non-trivial paths.  We claim that the degree of $p^*$ in $E_{\oplus}$ must be $0$. Then $p^*$ can be regarded as an additional trivial alternating path.  This trivial alternating path, together with the $|\Pi^+|$ non-trivial paths in $E_{\oplus}$ and the trivial paths in $\Pi^0$, form a larger feasible solution for $G_M(S, T)$.

To see why $p^*$ cannot have degree $2$ in $E_{\oplus}$, we should interpret $E_{\oplus}$ as the symmetric difference of two maximum matchings.  A node with degree $2$ in $E_{\oplus}$ must be matched in both matchings.  However, since $p^* \in \src_{\Pi^+} \subseteq S$, $p^*$ is not matched by $M$. Therefore, $p^*$ cannot have degree $2$ in $E_{\oplus}$.  This completes the proof.  
\end{proof}

The above proofs immediately imply the following.
\begin{corollary}
	\label{coro:augment}
	Let $M$ be a maximum matching of $G$.  Let $S$ be a subset of the players not matched by $M$.  Let $T$ be a subset of $P$.  Let $\Pi$ be a feasible solution for $G_M(S, T)$.   If $\Pi$ is not optimal, there exists a path $\pi$ in $G_{M\oplus \Pi}$ from $S \setminus \src_{\Pi}$ to $T \setminus \sink_\Pi$.  In polynomial time, we can augment $\Pi$ to a larger solution $\Pi'$ for $G_M(S, T)$ such that $|\Pi'| = |\Pi| + 1$, $\src_{\Pi'} = \src_{\Pi} \cup \{\src_{\pi}\}$, and $\sink_{\Pi'} = \sink_{\Pi} \cup \{\sink_{\pi}\}$.  Moreover, the collection of nodes used by $\Pi'$ is a subset of those used by $\Pi$ and $\pi$.
\end{corollary}
\begin{corollary}
	Let $M$ be a maximum matching of $G$.  Let $S$ be a subset of the players not matched by $M$.  Let $T$ be a subset of $P$.  An optimal solution for $G_M(S, T)$ can be computed in polynomial time.
\end{corollary}

The following lemma states that if we get one more node-disjoint path by including $p$ in $T$, so can we by including $p$ to any subset of $T$.

\begin{lemma}[\cite{CM18a}]
	\label{lem:submod}
	Let $M$ be a maximum matching of $G$.  Let $S$ be any subset of the unmatched players. Let $T$ be any subset of $P$.  Let $p$ be an arbitrary player in $P$. If $f_M(S,T \cup \{p\}) = f_M(S,T) + 1$, then for every $T' \subseteq T$, $f_M(S,T'\cup \{p\}) = f_M(S,T') + 1$.
\end{lemma}
\begin{proof}
	Suppose that $f_M(S,T \cup \{p\}) = f_M(S,T) + 1$.  Obviously, $p \notin T$.  Let $T'$ be an arbitrary subset of $T$.  Let $\Pi_1$ be an optimal solution for $G_M(S,T')$.  $p \notin \mathit{sink}_{\Pi_1}$.  Note that $\Pi_1$ is also a feasible solution for $G_M(S,T \cup \{p\})$.  Let $\Pi_2$ be an optimal solution for $G_M(S,T \cup \{p\})$ obtained by augmenting $\Pi_1$ (using Corollary~\ref{coro:augment}).  Then, $\mathit{sink}_{\Pi_1} \subseteq \mathit{sink}_{\Pi_2}$.  If $p \in \mathit{sink}_{\Pi_2}$, then $(\mathit{sink}_{\Pi_1} \cup \{p\}) \subseteq \mathit{sink}_{\Pi_2}$, implying that there are $|\Pi_1| + 1 = f_M(S,T') + 1$ node-disjoint paths from $S$ to $T' \cup \{p\}$, and thus establishing the lemma.  If $p \not\in \mathit{sink}_{\Pi_2}$, then $\Pi_2$ is a feasible solution for $G_M(S,T)$.  But then $f_M(S,T \cup \{p\}) = |\Pi_2| \leq f_M(S,T)$, a contradiction to the assumption.
\end{proof}
\section{Omitted Proofs in the Analysis of the Approximation Algorithm}
\label{apd:proofs}

\subsection{Proof of Lemma~\ref{lem:canon}}
	We first compute an optimal solution $\Pi_1$ for $G_M(B_1,I)$.  For $j = 2,\ldots,\ell-1$, we compute an optimal solution $\Pi_j$ for $G_M(B_{\leq j},I)$ by successively augmenting $\Pi_{j-1}$ using Corollary~\ref{coro:augment}.  Augmentation ensures that $\src_{\Pi_{j-1}} \subseteq \src_{\Pi_j}$.  Therefore, we inductively maintain the property that for all $i \in [1,j]$, $\Pi_j$ contains $|\Pi_i| = f_M(B_{\leq i},I)$ node-disjoint paths from $B_{\leq i}$ to $I$.  In the end, we obtain $\Pi_{\ell-1}$.  By invariant~1 in Table~\ref{tb:invar}, $\sink_{\Pi_{\ell-1}} = I$.  We obtain the canonical decomposition and the canonical solution as follows: for every $i \in [1,\ell-1]$, let $\Gamma_i$ be the subset of paths in $\Pi_{\ell-1}$ starting at $B_i$, and let $I_i = \sink_{\Gamma_i}$, let ${\cal I}_i$ be the subset of edges in $\cal I$ that cover the players in $I_i$.
				
\subsection{Proof of Lemma~\ref{lem:invar}}

We break the proof into two parts, handling {\sc Build} and {\sc Collapse} separately.
				
\begin{lemma}
	\label{lem:build}
	{\sc Build} maintains the invariants in Table~\ref{tb:invar}.
\end{lemma}
\begin{proof}
	Suppose that the invariants hold before {\sc Build}.  Let $L_{\ell+1}$ be the layer newly constructed by {\sc Build}.  We show that the invariants hold after the construction of $L_{\ell+1}$. 

	Consider invariant~1 and the moment when {\sc Build} is about to add an edge $(p, D)$ to $I$.
	By definition of addable players, we have 
	\[
		f_M(B_{\leq \ell}, A_{\ell+1} \cup I \cup \{p\}) = f_M(B_{\leq \ell}, A_{\ell+1} \cup I) + 1.
	\]
	Lemma~\ref{lem:submod} implies that
	\[
		f_M(B_{\leq \ell}, I \cup \{p\}) = f_M(B_{\leq \ell}, \cup I) + 1.
	\]
	Hence, whenever we add an edge to ${\cal I}$, both of the left hand side and right hand side of invariant~1 increase by $1$.  Other steps obviously have no effect on this invariant.	

	Consider invariant~2.  {\sc Build} does not change any old layer. Also, {\sc Build} does not delete any edge from $\cal I$.  Therefore, for all $i \in [1,\ell-1]$, $f_M(B_{\leq i},A_{ i+1} \cup I)$ cannot decrease and the inequality still holds.  By construction, {\sc Build} sets $d_{\ell+1} := f_M(B_{\leq \ell}, A_{\ell+1} \cup I)$.  So invariant 2 is preserved.

	Consider invariant~3 and~4. For $i\in [1,\ell]$, the inequalities continue to hold because {\sc Build} does not change any old layer.  It remains to show that the inequalities hold for $i = \ell + 1$.  By definition, $z_{\ell + 1} = |A_{\ell+1}|$, so invariant 3 holds.  At the beginning of {\sc Build}, ${\cal A}_{\ell + 1} = \emptyset$. 	Whenever we add an edge $(p, D)$ to ${\cal A}_{\ell + 1}$, the value of $f_M(B_{\leq \ell}, A_{\ell+1} \cup I)$ increases by $1$ because $p$ is an addable player.  Therefore, at the end, the value of $f_M(B_{\leq \ell}, A_{\ell+1} \cup I)$ is at least $|{\cal A}_{\ell + 1}|$.
	\[
		d_{\ell+1} = f_M(B_{\leq \ell}, A_{\leq \ell+1} \cup I) \geq |{\cal A}_{\ell + 1}| =  z_{\ell + 1}. 
	\]
	So Invariant~4 is preserved.
\end{proof}

	To prove that {\sc Collapse} preserves the invariants in Table~\ref{tb:invar}, we need the following result.
	
\begin{lemma}
	\label{lem:common-sol}
	Let $({\cal I}_1, \ldots, {\cal I}_{\ell-1})$ be a canonical decomposition of $\cal I$. Let $\Gamma_1 \cup  \cdots \cup \Gamma_{\ell-1}$ be a corresponding canonical solution for $G_M(B_{\leq \ell-1}, I)$.  For every $i \in [1, \ell - 2]$, $G_M(B_{\leq i}, A_{i+1} \cup I_{\leq i})$ and $G_M(B_{\leq i}, A_{i+1} \cup I)$ share a common optimal solution that is node-disjoint from $\Gamma_{i+1}\cup\cdots\cup \Gamma_{\ell-1}$. 
\end{lemma}
\begin{proof}
	Fix some $i \in [1,\ell-2]$. We compute an optimal solution for $G_M(B_{\leq i}, A_{i+1} \cup I)$ as follows. Let $\Pi_0 = \Gamma_1 \cup \cdots\cup \Gamma_i$.  $\Pi_0$ is a feasible solution for $G_M(B_{\leq i}, A_{i+1} \cup I)$ as paths in it go from $B_{\leq i}$ to $I_{\leq i}$.  We iteratively augment $\Pi_0$ using Corollary~\ref{coro:augment}.  Let $(\Pi_0, \Pi_1, \ldots, \Pi_k)$ be the intermediate solutions obtained during the repeated augmentations, where $\Pi_k$ is the optimal solution for $G_M(B_{\leq i}, A_{i+1} \cup I)$ we obtain at the end.
			
	We show that $\Pi_k$ is node-disjoint from $\Gamma_{> i} = \Gamma_{i+1}\cup\cdots\cup \Gamma_{\ell-1}$.  The set of sinks of $\Gamma_{> i}$ is $I_{i+1} \cup \cdots I_{\ell-1}$, which is denote as $I_{> i}$.
			
	For $j \in [0, k-1]$, let $\pi_j$ be the path in $G_{M\oplus \Pi_j}$ that is used to augment $\Pi_j$ to $\Pi_{j+1}$.  We claim that every $\pi_j$ is node-disjoint from $\Gamma_{> i}$.  Suppose not.  Let $\pi_{j^*}$ be the first such path that shares some node with $\Gamma_{> i}$.  Since $\Pi_0$ and every $\pi_j$ with $j < j^*$ are node-disjoint from $\Gamma_{> i}$,  $\Pi_{j^*}$ must be node-disjoint from $\Gamma_{> i}$.  Hence, $\Gamma_{> i}$ remains to be paths in $G_{M \oplus \Pi_{j^*}}$.  Since $\pi_{j^*}$ share some node with $\Gamma_{> i}$, we can construct a path $\pi^*$ in $G_{M \oplus \Pi_{j^*}}$ as follows:  start with $\src_{\pi_{j^*}}$, follow $\pi_{j^*}$, switch at the first common node of $\pi_{j^*}$ and $\Gamma_{> i}$, and follow a path in $\Gamma_{> i}$ to some player $p^*$ in $I_{> i}$.  If we augment $\Pi_{j^*}$ using $\pi^*$, we will get a feasible solution $\Pi^*$ for $G_M(B_{\leq i}, A_{i+1} \cup I)$.  All the players in $I_{\leq i} = \sink_{\Pi_0}$ are sinks of $\Pi^*$.  The player $p^* \in I_{>i}$ is also a sink of $\Pi^*$.    Consequently, there are $|I_{\leq i}| + 1$ node-disjoint paths in $\Pi^*$ from $B_{\leq i}$ to  $I$, which implies that $f_M(B_{\leq i}, I) \geq |I_{\leq i}| + 1$.  This contradicts the definition of canonical decomposition, thereby establishing our claim.
			
	By our claim, all $\pi_j$'s are nodes-disjoint from $\Gamma_{> i}$.  So the repeated augmentations to produce $\Pi_k$ do not generate any node-sharing with $\Gamma_{> i}$.  That is, $\Pi_k$ is node-disjoint from $\Gamma_{> i}$.

	Since $\Pi_k$ is node-disjoint from $\Gamma_{> i}$, no player from $I_{> i}$ can be a sink of $\Pi_k$. Therefore, $\Pi_k$ is also an optimal solution for $G_M(B_{\leq i}, A_{i+1} \cup I_{\leq i})$.
\end{proof}
	
\begin{lemma}
	{\sc Collapse} maintains invariants~1--4 in Table~\ref{tb:invar}.
\end{lemma}
\begin{proof}
	It suffices to show that invariants~1--4 are preserved after collapsing the lowest collapsible layer $L_t$ in steps~\ref{step:collapse-2}--\ref{step:collapse-4}.

	Consider invariant~1.  Since all the layers above the $t$-th layer are removed and we set ${\cal I}: = {\cal I}_{\leq t-1}$, it suffices to show $f_M(B_{\leq t-1}, I_{\leq t-1}) = |I_{\leq t-1}|$.  Let $\Gamma_{\leq t-1} = \Gamma_1 \cup \cdots \cup \Gamma_{t-1}$.  $\Gamma_{\leq t-1}$ are node-disjoint paths in $G_M$ from $B_{\leq t-1}$ to every player in $I_{\leq t-1}$ .  Obviously steps~\ref{step:collapse-2} and~\ref{step:collapse-4} do not affect $\Gamma_{\leq t-1}$.  Step~\ref{step:collapse-3} updates $M$, and may affect the paths in $G_M$.  The paths in $\Gamma_{\leq t-1}$ are node-disjoint from those in $\Gamma_t$, so after updating $M$ using $\Gamma_t$, the paths in $\Gamma_{\leq t-1}$ remain to be alternating paths with respect to the updated $M$.  Hence, after step 3, $\Gamma_{\leq t-1}$ is still a set of paths in $G_M$ from $B_{\leq t-1}$ to every player in $I_{\leq t-1}$.  It certifies that \[
		f_M(B_{\leq t-1}, I_{\leq t-1}) = |I_{\leq t-1}|.
	\]

	Consider invariant~2.  Since $L_t$ is going to be the topmost layer, we only need to show that these inequalities hold for $i \in [1, t-1]$ after steps~\ref{step:collapse-2}--\ref{step:collapse-4}.  Fix some $i \in [1, t-1]$.  By Lemma~\ref{lem:common-sol}, $G_M(B_{\leq i-1}, A_{i+1} \cup I_{\leq i})$ and $G_M(B_{\leq i}, A_{i+1} \cup I)$ share a common optimal solution that is node-disjoint from $\Gamma_{t}$.   Let $\Pi^*$ be that optimal solution.  By inductive hypothesis, we have $|\Pi^*| \geq d_{i+1}$.  $\Pi^*$ is not affected by step 2 because all its sinks belong to $A_{\leq i+1} \cup I_{\leq i}$. Similar to the proof of invariant~2, since $\Pi^*$ is node-disjoint from $\Gamma_t$, it is a set of paths in $G_M$ after step~3 updates $M$ using $\Gamma_t$. Thus, $\Pi^*$ certifies that 
	\[
		f_M(B_{\leq i}, A_{i+1} \cup I)\geq d_{i+1}
	\]
	at the end of step 3.  In step~4.1, the removal of edges from ${\cal A}_t$ may decrease the value of $f_{M}(B_{\leq t-1},A_{t} \cup I)$, but it does not change $f_{M}(B_{\leq i},A_{i+1} \cup I)$ for $i \in [1,t-2]$.  Suppose that $f_{M}(B_{\leq t-1},A_{t} \cup I)$ decreases after removing an edge $(p,D)$ from ${\cal A}_t$, that is, 
	\[
		f_{M}(B_{\leq t-1}, A_{t} \cup I)
		= f_{M}(B_{\leq t-1}, (A_{t} \setminus \{p\} ) \cup I) + 1.
	\] 
	Then when we reach step~4.2, Lemma~\ref{lem:submod} implies that 
		\[f_{M}(B_{\leq t-1},I \cup \{p\}) = f_{M}(B_{\leq t-1},I) + 1,\] 
	so step~4.2 will add $p$ to $I$.  Afterwards, $f_{M}(B_{\leq t-1},A_{t} \cup I)$ returns to its value prior to the removal of $(p,D)$ from ${\cal A}_t$.  As a result, invariant~2 holds after step~\ref{step:collapse-4}.

	Invariant 3 and 4 hold because {\sc Collapse} neither grows any ${\cal A}_i$ nor changes any $d_i$ and $z_i$ for $i \in [1,t]$.
\end{proof}
	
\subsection{Proof of Lemma~\ref{lem:non-collapse}}
	Consider (i).  By invariant~1 in Table~\ref{tb:invar}, $f_M(B_{\leq \ell-1}, I) = |I|$.  Hence, in the canonical decomposition of ${\cal I}$, we have $I_{\leq \ell-1} = I$.  If $|I_{\leq \ell-1}| = |I| > \mu |B_{\leq \ell -1}|$, by the pigeonhole principle, there exists an index $i \in[1,\ell-1]$ such that $|I_i| > \mu|B_i|$.  But then layer $L_i$ is collapsible, a contradiction.
		
	Consider (ii).  Assume to the contrary that there exists $i \in [1,\ell]$ such that $|A_{i}| <  z_i - \mu|B_{\leq i-1}|$.  Equivalently, $z_i > |A_{i}| + \mu|B_{\leq i-1}|$.  By invariants~2 and~4 in Table~\ref{tb:invar}, $f_M(B_{\leq i-1}, A_{i} \cup I) \geq d_{i} \geq z_i > |A_{i}| + \mu|B_{\leq i-1}|$.  Therefore, any optimal solution for $G_M(B_{\leq i-1},A_{i} \cup I)$ contains at least $\mu|B_{\leq i-1}| + 1$ node-disjoint paths from $B_{\leq i-1}$ to $I$.  It follows that $f_M(B_{\leq i-1},I) \geq \mu|B_{\leq i-1}| + 1$.  By the definition of canonical decomposition, $|I_{\leq i-1}| = f_M(B_{\leq i-1},I) \geq \mu|B_{\leq i-1}| + 1$. By the pigeonhole principle, there exists some $j \in [1,i-1]$ such that $|I_j| > \mu|B_j|$.  But then layer $L_j$ is collapsible, a contradiction.

\subsection{Proof of Claim~\ref{cl:residual-player}}
Consider (i).  Recall that $\Pi$ is an optimal solution for $G_M(B_{\leq \ell}, A_{\ell+1} \cup I)$.  The players in $B_{\leq \ell}\setminus \src_\Pi$ cannot be sinks of $\Pi$ because they have in-degree $0$ in $G_M$ and cannot be reached by paths not starting with them.  In $G_{M \oplus \Pi}$, the players in $\sink_{\Pi}$ have in-degree 0, and hence cannot be reached from the other players.  In particular, they cannot be reached from the players in $B_{\leq \ell}\setminus \src_\Pi$.  Let $p$ be an arbitrary player in $P$.  By definition, $p$ is reachable in $G_{M\oplus \Pi}$ from $B_{\leq \ell}\setminus \src_\Pi$, so $p \notin \sink_{\Pi}$.  Note that $\Pi$ is a feasible solution for $G_M(B_{\leq \ell}, A_{\ell+1} \cup I \cup \{p\})$, and $\Pi$ cannot be optimal by Lemma~\ref{lem:max-paths} because there is a path in $G_{M \oplus \Pi}$ from $B_{\leq \ell}\setminus \src_\Pi$ to $p \notin \sink_{\Pi}$.  In other words,
\[
		f_M(B_{\leq \ell}, A_{\ell+1} \cup I \cup \{p\}) > f_M(B_{\leq \ell}, A_{\ell+1} \cup I).
\]
So $p$ is addable.

(ii) is easy to prove because in $G_{M \oplus \Pi}$, a player has degree $0$ if it is not matched by $M\oplus\Pi$, and degree 1 otherwise.

To prove (iii), it suffices to show that every fat resource in $R^+_f$ must be matched by $M\oplus\Pi$.  Let $r_f$ be an arbitrary fat resource in $R^+_f$.  By definition, $r_f$ is reachable in $G_{M\oplus \Pi}$ from some player $p^*$ in $B_{\leq \ell}\setminus {\mathit src}_\Pi$.  Note that $p^*$ is not matched by $M\oplus \Pi$.  If $r_f$ is not matched neither, then the path between $p^*$ and $r_f$ would be an alternating path from an unmatched player to an unmatched fat resource, which can be used to increase the size of $M\oplus \Pi$~\cite{HK73}.  But $M \oplus \Pi$ is already a maximum matching of $G$, contradiction.  Hence, $r_f$ must be matched.  This completes the proof.
\section{Integrality Gap: Proof of Theorem~\ref{thm:gap}}
\label{apd:gap}

We present a tighter analysis for the local search algorithm in~\cite{AFS12}.  We shows that an allocation can be computed such that every player receives at least $\lambda = \frac{26}{99}$ worth of resources.  Recall that the optimal value $T^*$ of the configuration LP is assumed to be $1$.  This proves that the integrality gap is at most $\frac{99}{26}\approx 3.808$.  The computation time is not known to be polynomial though.

\subsection{The local search algorithm}
\label{sec:afs-alg}

We present the algorithm in a way slightly different from that in~\cite{AFS12}, in order to show the similarity between this algorithm and the approximation algorithm in Section~\ref{sec:approx}.  Let $M$ and ${\cal E}$ be the current maximum matching of $G$ and the current set of thin edges maintained by the algorithm, respectively.  Let $p_0$ be a player not yet covered by $M \cup {\cal E}$.

The algorithm maintains a stack of tuples $\Sigma = [(a_1,{\cal B}_1), (a_2,{\cal B}_2), \cdots]$, where $a_i$ is an addable edge and ${\cal B}_i$ is the set of blocking edges of $a_i$.  We will give the definitions of addable edges and blocking edges soon.  For $i < j$, $(a_i, {\cal B}_i)$ is pushed into $\Sigma$ before $(a_j, {\cal B}_j)$.   For simplicity, we define $a_1 = \mathit{null}$ and ${\cal B}_1 = \{(p_0, \emptyset)\}$.  We use $\ell$ to denote the length of $\Sigma$.  We use $R(\Sigma)$ to denote the set of thin resources covered by $\{a_1, \ldots, a_{\ell}\} \cup {\cal B}_1 \cdots \cup {\cal B}_{\ell}$.  For $i \in [1, \ell]$, $B_i$ denotes the set of players covered by ${\cal B}_i$, ${\cal B}_{\leq i}$ denotes $\bigcup_{j = 1}^{i} {\cal B}_j$, and $B_{\leq i}$ denotes $\bigcup_{j = 1}^{i} B_j$.

The stack $\Sigma$ is built inductively.  Initially, $\Sigma = [(a_1, {\cal B}_1)]$.  Consider the construction of $(a_{\ell+1}, {\cal B}_{\ell+1})$.  Recall that $G_M$ is the directed graph obtained from $G$ by orienting edges in $G$ from $r_f$ to $p$ if $\{p, r_f\}\in M$ and $p$ to $r_f$ otherwise. Also recall that, for any $w > 0$, a thin edge $(p, D)$ is \emph{$w$-minimal} if $v[D] \geq w$ and $v[D'] < w$ for any $D' \subsetneq D$. 

\begin{definition}
	Given the current partial allocation $M\cup {\cal E}$ and the current stack $\Sigma = [(a_1, {\cal B}_1), \ldots, (a_\ell, {\cal B}_{\ell})]$, a player $p$ is \textbf{addable} if, in $G_M$, there is a path to $p$ from some player in $B_{\leq \ell}$. 
\end{definition}

\begin{definition}
	Given the current partial allocation $M \cup {\cal E}$ and the current stack $\Sigma = [(a_1, {\cal B}_1), \ldots, (a_\ell, {\cal B}_{\ell})]$, a thin edge $(p, D)$ is \textbf{addable} if (i)~$p$ is addable, (ii)~$D \cap R(\Sigma) = \emptyset$, and (iii)~$(p, D)$ is $\lambda$-minimal.  For an addable edge $a$, an edge $b$ in $\cal E$ is a \textbf{blocking} edge of $a$ if $b$ shares some common resource with $a$.  If an addable edge has no blocking edge, it is {unblocked}; otherwise, it is {blocked}.
\end{definition}

The construction of $(a_{\ell+1},{\cal B}_{\ell+1})$ is specified in the following routine {\sc Build}.
\begin{quote}
{\sc Build}$\left(M, {\cal E}, \Sigma, \ell\right)$
\begin{enumerate}
\item 	Arbitrarily pick an addable edge $a_{\ell+1}$.
\item 	${\cal B}_{\ell + 1} := \{\text{$e \in {\cal E}$ : $e$ is a blocking edge of $a_{\ell+1}$}\}$.
\item 	Append $(a_{\ell+1}, {\cal B}_{\ell + 1})$ to $\Sigma$. Set $\ell := \ell + 1$.
\end{enumerate}
\end{quote}

Once some addable edge $a_i$ in $\Sigma$ is unblocked, i.e., ${\cal B}_i = \emptyset$, the following routine is invoked to update $\Sigma$, $M$, and $\cal E$.
\begin{quote}
{\sc Contract}$\left(M, {\cal E}, \Sigma, \ell\right)$
\begin{enumerate}
\item 	
	Let $a^*$ be the unblocked addable edge with the smallest index. 
	Let $p_{a^*}$ be the player covered by $a^*$.

\item   
	Let $t$ be the smallest index such that, in $G_M$, there is a path $\pi$ to $p_{a^*}$ from a player $p_t$ covered by some blocking edge $b_t$ in ${\cal B}_t$. Note that $\pi$ is an alternating path with respect to $M$.~\label{step:contract-2}

\item
	Delete all the tuples $(a_i, {\cal B}_i)$'s with $i > t$. Set $\ell := t$.

\item Update $M$ and ${\cal E}$ as the following.
	\begin{enumerate}[{4.}1]
		\item Update $M$ using $\pi$, i.e., $M := M \oplus \pi$,
		\item add $a^*$ to ${\cal E}$ and release $b_t$, i.e., ${\cal E} : = \left({\cal E} \setminus \{b_t\}\right) \cup \{a^*\}$.
	\end{enumerate}

\item 	
	If $t = 1$, step~3 already matches $p_0$, so the algorithm terminates.  If $t \geq 2$, as $b_t$ is released from ${\cal E}$, it is no longer a blocking edge. Set ${\cal B}_t := {\cal B}_t \setminus \{b_t\}$.
\end{enumerate}
\end{quote}

The algorithm keeps calling {\sc Contract} until no unblocked addable edge remains in $\Sigma$.  Then it calls {\sc Build} to grow the stack $\Sigma$ again.  It alternates between calling {\sc Build} and {\sc Contract} until $p_0$ is covered.

\subsection{Analysis}

We prove an invariant maintained by the algorithm. It is implicitly used in step~\ref{step:contract-2} of {\sc Contract} to guarantee that $t$ always exists.
\begin{lemma}
\label{lem:gap-invar}
Let $M \cup {\cal E}$ be the current partial allocation.  Let $\Sigma = [(a_1, {\cal B}_1), \ldots, (a_\ell, {\cal B}_\ell)]$ be the current stack.  For all $i \in [2, \ell]$, there is always a path in $G_M$ that goes from some player in $B_{\leq i-1}$ to the player covered by $a_{i}$.
\end{lemma}
\begin{proof}
	We prove the lemma by induction.  Initially, $\Sigma = \{(a_1, {\cal B}_1)\}$ and $\ell = 1$.  So the invariant trivially holds.
	
	We show that {\sc Build} preserves the invariant.  Let $(a_{\ell + 1}, {\cal B}_{\ell + 1})$ be the tuple newly constructed by {\sc Build}.  For $i \in [1, \ell]$, since {\sc Build} does not change the first $\ell$ tuples in $\Sigma$ nor $M \cup {\cal E}$, by the inductive hypothesis, the player covered by $a_{i}$ is always reachable from some player in $B_{\leq i-1}$.  The player covered by $a_{\ell + 1}$ must be addable. By definition, there is a path in $G_M$ that goes from some player in $B_{\leq \ell}$ to the player covered by $a_{i}$.

	We show that {\sc Contract} also preserves the lemma.  Since all the tuples with index greater than $t$ are deleted in step~3, we only need to verify the invariant for the remaining $t$ tuples.  Let $p_{a^*}$, $t$, $p_t$, $\pi$ be defined as in the description of {\sc Contract}.  By our choice of $t$, players in $B_{\leq t-1}$ cannot reach $p_{a^*}$ by any path in $G_M$.  Hence, any path $\pi'$ starting with some player in $B_{\leq t-1}$ must be node-disjoint from $\pi$, since otherwise, we can find a path from $B_{\leq t-1}$ to $p_{a^*}$ by first following $\pi'$, switching at the common node of $\pi'$ and $\pi$, and then following $\pi$.  Therefore, paths originating from $B_{\leq t-1}$ are not affected by the operation $M\oplus \pi$.
\end{proof}

Lemma~\ref{lem:gap-nonstuck} below shows that the algorithm never gets stuck.  Recall that $R(\Sigma)$ is the set of thin resources covered by the thin edges in $\Sigma$.  Also recall that the restricted max-min allocation problem can be modeled as a configuration LP.  Consider the dual of the configuration LP.  Intuitively, we show that if the lemma does not hold, then some lower bound for the total dual value of the thin resources in $R(\Sigma)$ would exceed its upper bound, which is a contradiction.  Asadpour et al.~\cite{AFS12} set the dual value of a thin resource $r$ to be $v_r$, and then used a worst-case upper bound and a worst-case lower bound for the total value of the resources in $R(\Sigma)$.  Their proof works only for $\lambda \leq \frac{1}{4}$.  We observe that the worst-case upper bound and the worst-case lower bound used in~\cite{AFS12} cannot occur simultaneously.  More specifically, the worst-case upper bound occurs only when all the thin resources have values nearly $\lambda$, while the worst-case lower bound occurs only when all the thin resources have values nearly $0$.  Our approach is to magnifying the dual value of the thin resources with small $v_r$.  It helps us to derive a better lower bound without deteriorating the upper bound.  Hence, we can obtain a proof working for a larger $\lambda$.

\begin{lemma}
\label{lem:gap-nonstuck}
	Let $M \cup {\cal E}$ be the current partial allocation.  Let $\Sigma = [(a_1, {\cal B}_1), \ldots, (a_\ell, {\cal B}_\ell)]$ be the current stack.  If $\Sigma$ is non-empty, either some addable edge in $\Sigma$ is unblocked or there is an addable edge to be added to $\Sigma$.
\end{lemma}
\begin{proof}
	We prove by contradiction.  Assume that all addable edges in $\Sigma$ are blocked, i.e., $|{\cal B}_i| \geq 1$ for all $i\in [1, \ell]$, and that there is no addable edge to be added to $\Sigma$.  Recall that the optimal value of the configuration LP is assumed to be $1$.  We show that the dual of the configuration LP $\mathit{CLP}(1)$ is unbounded, which implies the contradiction that $\mathit{CLP}(1)$ is infeasible.

	Consider the following solution for the dual of $\mathit{CLP}(1)$. Let $P^+$ be the set of players that are reachable in $G_M$ from some player in $B_{\leq \ell}$.  Let $R^+_f$ be the set of fat resources that are reachable in $G_M$ from some player in $B_{\leq \ell}$.  Recall that $R(\Sigma)$ is the set of thin resources covered by thin edges in $\Sigma$.  We set $\lambda = \frac{26}{99}$.  For every $p\in P$ and every $r\in R$, set the dual variable $y^*_p$ and $z^*_r$ as follows.

	\begin{alignat*}{4}
		y^*_p &= \left\{
				\def\arraystretch{1.5}
				\begin{array}{ll}
					1 - \frac{21}{26}\lambda &\text{if $p \in P^+$,}\\
					0 &\text{otherwise.}
				\end{array}
		\right. \qquad
		&z^*_r &= \left\{
				\def\arraystretch{1.5}
				\begin{array}{ll}
					1 - \frac{21}{26}\lambda 
						&\text{if $r\in R^+_f$,}\\
					\frac{3\lambda}{2\lambda + v_r}v_r 
						&\text{if $r \in R(\Sigma)$ and $v_r \in (0,\frac{\lambda}{2})$}\\
					\frac{3\lambda}{3\lambda - v_r}v_r 
						&\text{if $r \in R(\Sigma)$ and $v_r \in [\frac{\lambda}{2},\frac{3\lambda}{4})$}\\
					\lambda 
						&\text{if $r \in R(\Sigma)$ and $v_r \in [\frac{3\lambda}{4},\lambda)$}\\
					0
						&\text{otherwise.}
				\end{array}
				\right.
	\end{alignat*}

	Figure~\ref{fig:zr} plots $z^*_r$ and $z^*_r/v_r$ versus $v_r$ for a thin resource $r$ in $R(\Sigma)$.  In future analysis, we will draw some conclusions directly from Figure~\ref{fig:zr} without giving a formal proof.  One can verify these conclusions by easy numeric calculation.

	\begin{claim}[Feasibility]
		\label{cl:gap-dual-feasible}
		$y^*_p \leq \sum_{r \in C}z^*_r$ for any $p\in P$ and any $C \in {\cal C}_p(1)$.
	\end{claim}

	\begin{claim}[Positive Objective Function Value]
		\label{cl:gap-dual-positive}
		$\sum_{p\in P}y^*_p - \sum_{r\in R}z^*_r > 0$.
	\end{claim}

	Suppose that Claims~\ref{cl:gap-dual-feasible} and~\ref{cl:gap-dual-positive} hold. $(\{y^*_p\}_{p\in P}, \{z^*_r\}_{r\in R})$ is a feasible solution for the dual, so is $(\{\alpha y^*_p\}_{p\in P}, \{\alpha z^*_r\}_{r\in R})$ for any $\alpha > 0$.  As $\alpha$ goes to infinity, the objective function value goes to infinity.  Therefore, the dual is unbounded, a contradiction.
\end{proof}

\begin{figure}
	\centering
	\begin{subfigure}[b]{0.45\linewidth}
		\centering
		\begin{tikzpicture}[xscale = 4, yscale = 2]
			\draw [<->] (0,1.3) -- (0,0) -- (1.3,0);
			\draw[black, thick, domain=0:0.5] plot (\x, {3* \x /(2 + \x)});
			\draw[black, thick, domain=0.5:0.75] plot (\x, {3* \x /(3 - \x)});
			\draw[black, thick, domain=0.75:1] plot (\x, 1);

			\node [left] at (0, 1.3) {$z^*_r$};
			\node [below right] at (1.3, 0) {$v_r$};
			\node [below left] at (0,0) {$0$};

			\draw [dashed] (0,0.6) -- (0.5,0.6) -- (0.5,0);
			\node [left] at (0, 0.6) {$\frac{3}{5}\lambda$};
			\node [below] at (0.5, 0) {$\frac{1}{2}\lambda$};

			\draw [dashed] (0,1) -- (0.75,1) -- (0.75,0);
			\node [left] at (0, 1) {$\lambda$};
			\node [below] at (0.75, 0) {$\frac{3}{4}\lambda$};
		\end{tikzpicture}
		\caption{$z^*_r$ versus $v_r$}
		\label{fig:z-v}
	\end{subfigure}
	\begin{subfigure}[b]{0.45\linewidth}
		\centering
		\begin{tikzpicture}[xscale = 4, yscale = 3.5]
			\draw [<->] (0,1.7) -- (0,0.9) -- (1.3,0.9);
			\draw[black, thick, domain=0:0.5] plot (\x, {3 /(2 + \x)});
			\draw[black, thick, domain=0.5:0.75] plot (\x, {3 /(3 - \x)});
			\draw[black, thick, domain=0.75:1] plot (\x, 1/\x);

			\node [left] at (0, 1.7) {$\frac{z^*_r}{v_r}$};
			\node [below right] at (1.3, 0.9) {$v_r$};
			\node [below left] at (0,0.9) {$0$};

			\node [left] at (0, 1.5) {$\frac{3}{2}$};

			\draw [dashed] (0,1.2) -- (0.5,1.2) -- (0.5,0.9);
			\node [left] at (0, 1.2) {$\frac{6}{5}$};
			\node [below] at (0.5, 0.9) {$\frac{1}{2}\lambda$};

			\draw [dashed] (0,4/3) -- (0.75,4/3) -- (0.75,0.9);
			\node [left] at (0, 4/3) {$\frac{4}{3}$};
			\node [below] at (0.75, 0.9) {$\frac{3}{4}\lambda$};

			\draw [dashed] (0,1) -- (1,1) -- (1,0.9);
			\node [left] at (0, 1) {$1$};
			\node [below] at (1, 0.9) {$\lambda$};
		\end{tikzpicture}
		\caption{$z^*_r/v_r$ versus $v_r$}
		\label{fig:z-v-v}
	\end{subfigure}
	\caption{Dual values for resources in $R(\Sigma)$}
	\label{fig:zr}
\end{figure}
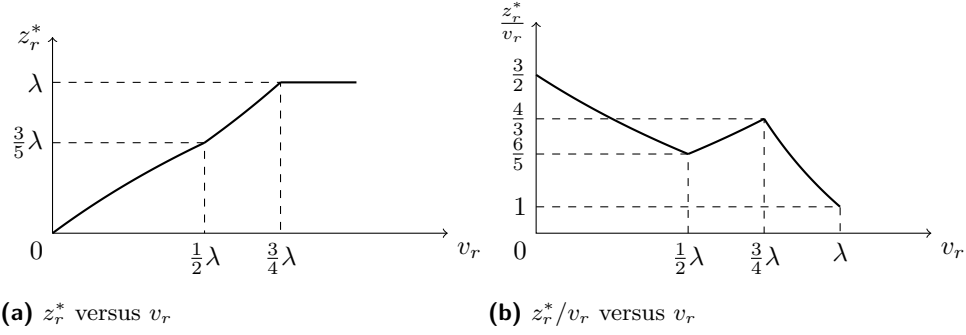

We give the proofs of Claims~\ref{cl:gap-dual-feasible} and~\ref{cl:gap-dual-positive} below.

\begin{claimproof}[Proof of Claim~\ref{cl:gap-dual-feasible}]
	We need to prove that $y^*_p \leq \sum_{r \in C}z^*_r$ for any $p\in P$ and any $C \in {\cal C}_p(1)$.  Consider any player $p \in P$ and any configuration $C \in {\cal C}_p(1)$.  If $p \notin P^+$, then $y^*_p = 0$, and the inequality holds because $z^*_r$ is non-negative.  Assume that $p \in P^+$. So $y^*_p = 1 - \frac{21}{26}\lambda$.  We prove that $\sum_{r \in C}z^*_r \geq 1 - \frac{21}{26}\lambda$ by a case analysis.
	
	\emph{Case 1}. $C$ contains a fat resource $r_f$.  Since $p \in P^+$, $p$ is reachable in $G_M$ from some player in $B_{\leq \ell}$.  Since player $p$ desires $r_f$, $G_M$ contains either an edge from $p$ to $r_f$ or an edge from $r_f$ to $p$.  In the former case, obviously $r_f$ is reachable in $G_M$ from some player in $B_{\leq \ell}$ .  In the latter case, $p$ must be a player matched to $r_f$ by $M$, so $p \notin B_{\leq \ell}$ as players in $B_{\leq \ell}$ are not matched by $M$.  Moreover, the edge from $r_f$ to $p$ is the only edge entering $p$ in $G_M$.  Any path entering $p$ must go though $r$.  Therefore, $r_f$ must be reachable from some player in $B_{\leq \ell}$.  In either case, $r_f \in R^+_f$. Therefore, $z^*_{r_f} = 1 -\frac{21}{26} \lambda$.  We have 
	\[
		\sum_{r \in C}z^*_r \geq z^*_{r_f} \geq 1 - \frac{21}{26}\lambda.
	\]

	\emph{Case 2}. $C$ contains only thin resources.  Since $p \in P^+$, $p$ is reachable from some player in $B_{\leq \ell}$.  So $p$ is an addable player.  However, by our assumption, no more addable edge can be added to $\Sigma$.  Thus, the total value of thin resources in $C\setminus R(\Sigma)$ must be less than $\lambda$, since otherwise there would be an addable edge formed by $p$ and the thin resources in $C\setminus R(\Sigma)$.  Given $\lambda = \frac{26}{99}$,
	\begin{equation}
		\sum_{r \in C\cap R(\Sigma)}v_r \geq \sum_{r\in C}v_r - \sum_{r \in C\setminus R(\Sigma)}v_r > 1 - \lambda = \frac{73}{26}\lambda.
		\label{eq:gap-0}
	\end{equation}
	We prove that 
	\[
		\sum_{r \in C\cap R(\Sigma)}z^*_r \geq  3\lambda \stackrel{(\because \lambda = 26/99)}{=} 1 - \frac{21}{26}\lambda.
	\]
	by examining subcases 2.1 -- 2.4 below.

	\emph{Case 2.1}.  $C\cap R(\Sigma)$ contains at least three resources with their values in $[\frac{3\lambda}{4},\lambda)$.  Denote these three resources as $r_1$, $r_2$, and $r_3$.  We have that  
		\[
			\sum_{r \in C\cap R(\Sigma)}z^*_r  \geq z^*_{r_1} + z^*_{r_2} + z^*_{r_3}
			= 3\lambda.
		\]

	\emph{Case 2.2}. $C\cap R(\Sigma)$ contains exactly one resource with value in $[\frac{3\lambda}{4},\lambda)$.  Denote this resource as $r_1$.  Let $R' = (C \cap R(\Sigma)) \setminus \{r_1\}$.  
		\begin{equation}
			\sum_{r\in R'}v_r = \sum_{r\in C\cap R(\Sigma)}v_r - v_{r_1} \stackrel{\eqref{eq:gap-0}}{>} \frac{73}{26}\lambda - \lambda = \frac{47}{26}\lambda.
			\label{eq:gap-1}
		\end{equation}
		Every resource in $R'$ has value in the range $(0, \frac{3}{4}\lambda)$.  As illustrated in Figure~\ref{fig:zr}(b), $\frac{z^*_r}{v_r} \geq \frac{6}{5}$ for every $r \in R'$.  Therefore,
		\[
			\sum_{r\in R'}z^*_r \geq \frac{6}{5}\sum_{r\in R'}v_r \stackrel{\eqref{eq:gap-1}}{>} \frac{141}{65}\lambda.
		\]
		Then,
		\[
			\sum_{r \in C\cap R(\Sigma)}z^*_r = \sum_{r\in R'}z^*_r + z^*_{r_1} > \frac{151}{65}\lambda  + \lambda > 3\lambda
		\]

	\emph{Case 2.3}. $C\cap R(\Sigma)$ contains no resource with value in $[\frac{3\lambda}{4},\lambda)$.  As illustrated in Figure~\ref{fig:zr}(b),  $\frac{z^*_r}{v_r} \geq \frac{6}{5}$ for every $r$ with $v_r \in (0, \frac{3\lambda}{4})$.  Then,
		\[
			\sum_{r\in C\cap R(\Sigma)}z^*_r \geq \frac{6}{5}\sum_{r\in C\cap R(\Sigma)}v_r \stackrel{\eqref{eq:gap-0}}{>}\frac{6}{5}\cdot \frac{73}{26}\lambda > 3\lambda.
		\]

	\emph{Case 2.4}.  The only remaining case is that $C\cap R(\Sigma)$ contains exactly two resources with values in $[\frac{3\lambda}{4},\lambda)$.  Denote these two resources as $r_1$ and $r_2$. $r_1$ and $r_2$ together contribute $2\lambda$ to the sum $\sum_{r \in C \cap R(\Sigma)} z^*_r$.  Let $R' = (C \cap R(\Sigma)) \setminus \{r_1,r_2\}$.  To prove that $\sum_{r \in C\cap R(\Sigma)}z^*_r \geq  3\lambda$, it suffices to show that 
		\[
			\sum_{r\in R'}z^*_r \geq \lambda.
		\] 

	Let $V_0$ denote the multi-set of values of thin resources in $R'$.  Note that $v \in (0, \frac{3}{4}\lambda)$ for any $v \in V_0$, and that  
	\[
		\sum_{v \in V_0} v = \sum_{r\in R'}v_r = \left(\sum_{r \in C \cap R(\Sigma)} v_r \right) - v_{r_1}-v_{r_2}  \stackrel{\eqref{eq:gap-0}}{>} \frac{73}{26}\lambda - 2\lambda = \frac{21}{26}\lambda.
	\]
	Let $g(v) = \frac{3\lambda}{2\lambda + v}v$. Let $h(v) = \frac{3\lambda}{3\lambda - v}v$.  We have 
	\[
		\sum_{r \in R'}z^*_r = \sum_{v \in V_0 \cap (0,\frac{\lambda}{2})} g(v) + \sum_{v \in V_0 \cap [\frac{\lambda}{2}, \frac{3\lambda}{4})} h(v).
	\]  

	To derive a lower bound,  we will transform $V_0$ step by step to another multiset $V_2$ of values such that 
		\begin{enumerate}[(1)]
			\item $\sum_{v \in V_2 \cap (0,\frac{\lambda}{2})} g(v) + \sum_{v \in V_2 \cap [\frac{\lambda}{2}, \frac{3\lambda}{4})} h(v) \leq \sum_{v \in V_0 \cap (0,\frac{\lambda}{2})} g(v) + \sum_{v \in V_0 \cap [\frac{\lambda}{2}, \frac{3\lambda}{4})} h(v)$,
			\item $\sum_{v \in V_2} v= \frac{21\lambda}{26}$, and
			\item $V_2$ contains exactly two values, one belonging to $(0, \frac{\lambda}{2})$ and the other belonging to $[\frac{\lambda}{2}, \frac{3\lambda}{4})$. 
		\end{enumerate}  
	Then we give a lower bound for $\sum_{v \in V_2 \cap (0,\frac{\lambda}{2})} g(v) + \sum_{v \in V_2 \cap [\frac{\lambda}{2}, \frac{3\lambda}{4})} h(v)$.

	As illustrated in Figure~\ref{fig:zr}(a), both $g(v)$ and $h(v)$ are increasing functions of $v$.  Hence, if we decrease the values in $V_0$ to some smaller positive values, $\sum_{v \in V_0 \cap (0,\frac{\lambda}{2})} g(v) + \sum_{v \in V_0 \cap [\frac{\lambda}{2}, \frac{3\lambda}{4})} h(v)$ does not increase.  We keep decreasing the values in $V_0$ in an arbitrary fashion until $\sum_{v \in V_0} v= \frac{21\lambda}{26}$.  Let $V_1$ be the resulting multi-set.  All values in $V_1$ are in the range $(0, \frac{3\lambda}{4})$.

	If $V_1$ contains exactly one value in $(0,\frac{\lambda}{2})$, then $V_1$ must contain exactly one value in $[\frac{\lambda}{2}, \frac{3\lambda}{4})$ because $\sum_{v \in V_1} v= \frac{21\lambda}{26}$. Hence, $V_1$ meets condition (3), and it is the multi-set we desire.  We set $V_2 = V_1$.  Suppose that $V_1$ contains at least two values $a$ and $b$ in $(0,\frac{\lambda}{2})$.  
	If $a + b \leq \frac{\lambda}{2}$, then we merge $a$ and $b$ into $c = a + b$.  If $c < \frac{\lambda}{2}$, 
	\[
		g(c) = \frac{3\lambda}{2\lambda + c} = \frac{3\lambda}{2\lambda + a + b}(a+b) \leq \frac{3\lambda}{2\lambda + a }a + \frac{3\lambda}{2\lambda + b}b = g(a) + g(b).
	\] 
	If $c = \frac{\lambda}{2}$,
	\[
		h(c) = \frac{3\lambda}{3\lambda - c}c = \frac{3\lambda}{2\lambda + c}c = \frac{3\lambda}{2\lambda + a + b}(a+b) \leq \frac{3\lambda}{2\lambda + a }a + \frac{3\lambda}{2\lambda + b}b = g(a) + g(b).
	\]
	Hence, $\sum_{v \in V_1 \cap (0,\frac{\lambda}{2})} g(v) + \sum_{v \in V_1 \cap [\frac{\lambda}{2}, \frac{3\lambda}{4})} h(v)$ does not increase.  
	
	If $a + b > \frac{\lambda}{2}$, we replace them with $c = a + b - \frac{\lambda}{2}$ and $d = \frac{\lambda}{2}$.  Note that $c < \frac{\lambda}{2}$ because $a + b < \lambda$.  Also, $h(d) = \frac{3\lambda}{3\lambda - d}d = \frac{3\lambda}{2\lambda + d}d$.  One can verify that 
	\[
		g(c) + h(d) = \frac{3\lambda}{2\lambda +  c}c + \frac{3\lambda}{2\lambda + d}d
			= 6\lambda\left(1 - \frac{4\lambda^2 + (c + d)\lambda}{4\lambda^2 + 2(c+d)\lambda + cd}\right),
	\]
	and that 
	\[
		g(a) + g(b) = \frac{3\lambda}{2\lambda +  a}a + \frac{3\lambda}{2\lambda + b}b
			= 6\lambda\left(1 - \frac{4\lambda^2 + (a + b)\lambda}{4\lambda^2 + 2(a+b)\lambda + ab}\right).
	\]
	Since $a + b = c + d$ and $ab - cd = ab - (a+b)\frac{\lambda}{2} + \frac{\lambda^2}{4} = (\frac{\lambda}{2} - a)(\frac{\lambda}{2} - b) > 0$, 
	\[
		g(c) + h(d) < g(a) + g(b).
	\]
	We conclude that $\sum_{v \in V_1 \cap (0,\frac{\lambda}{2})} g(v) + \sum_{v \in V_1 \cap [\frac{\lambda}{2}, \frac{3\lambda}{4})} h(v)$ does not increase.  

	The above operation reduces the number of values in $(0,\frac{\lambda}{2})$ by exactly $1$ while preserving the sum of values.  Repeating the above gives a multi-set $V_2$ that meets  condition (3).

	Now we are ready to derive a lower bound for $\sum_{v \in V_2 \cap (0,\frac{\lambda}{2})} g(v) + \sum_{v \in V_2 \cap [\frac{\lambda}{2}, \frac{3\lambda}{4})} h(v)$.  
	Let $a \in [\frac{\lambda}{2}, \frac{3}{4}\lambda)$ be the larger value in $V_2$. Then the smaller value is $(c - a)$ where $c = \frac{21}{26}\lambda$.
	\begin{align*}
		&\sum_{v \in V_2 \cap (0,\frac{\lambda}{2})} g(v) + \sum_{v \in V_2 \cap [\frac{\lambda}{2}, \frac{3\lambda}{4})} h(v) \\
	= & g(c-a) + h(a) \\
		=& \frac{3\lambda (c-a)}{2\lambda + (c-a)}  + \frac{3\lambda a}{3\lambda - a}\\
				=& 3\lambda\left(1 - \frac{2\lambda}{2\lambda + c - a}+ \frac{3\lambda}{3\lambda - a} - 1\right)\\
				=& 3\lambda\left(- \frac{2\lambda}{2\lambda + c - a}+ \frac{3\lambda}{3\lambda - a}\right).
	\end{align*}
	One can verify that when $a \in (\frac{\lambda}{2}, \frac{3\lambda}{4})$
			\begin{equation*}
				\frac{d}{da}\left(g(c-a) + h(a)\right)
				= 3\lambda\left( -  \frac{2\lambda}{(2\lambda + c - a)^2} + \frac{3\lambda}{(3\lambda - a)^2}\right) > 0
			\end{equation*}
	So we get the minimum when $a = \frac{\lambda}{2}$.  Therefore,  we have
	\[
		\sum_{r\in R'}z^*_r \geq g(c-a) + h(a) \geq g(\frac{21}{26}\lambda - \frac{1}{2}\lambda) + h(\frac{1}{2}\lambda) = \lambda.
	\]
\end{claimproof}

Before proving Claim~\ref{cl:gap-dual-positive}, we first establish the following result.
\begin{claim}
	\label{cl:edge-dual-value}
	Let $e$ be a thin edge that appears in $\Sigma$. That is, $e$ is either an addable edge $a_i$ or a blocking edge in ${\cal B}_i$ for some $i$. Let $R_e$ be the set of thin resources covered by $e$. Let $r_0$ be the resource with the least value in $R_e$. Then, 
	\begin{equation*}
		\sum_{r\in R_e} z^*_r \leq \frac{3\lambda}{2} + \frac{z^*_{r_0}}{2}.
	\end{equation*}
\end{claim}
\begin{claimproof}
	If $v_{r_0} \geq \frac{\lambda}{2}$, all resources in $R_e$ have values at least $\frac{\lambda}{2}$.  Since $e$ is $\lambda$-minimal, $R_e$ contains exactly two thin resources, including $r_0$.  Let $r_1$ denote the other resource in $R_e$.
	\[
		\sum_{r\in R_e} z^*_r = z^*_{r_0} + z^*_{r_1} \leq z^*_{r_0} + \lambda
				\leq \frac{z^*_{r_0}}{2} + \frac{\lambda}{2} + \lambda = \frac{z^*_{r_0}}{2} + \frac{3\lambda}{2}.
	\]

	Suppose that $v_{r_0} < \frac{\lambda}{2}$.  Let $r_1$ be the resources with the largest value in $R_e$.  If $v_{r_0} \geq \lambda - v_{r_1}$, then $r_0$ and $r_1$ have a total value of at least $\lambda$.  Thus $R_e$ does not contain any other resource because $e$ is $\lambda$-minimal.  We get $\sum_{r\in R_e} z^*_r = z^*_{r_0} + z^*_{r_1} \leq
	\frac{z^*_{r_0}}{2} + \frac{3\lambda}{2}$ as before.
	Suppose that $v_{r_0} < \lambda - v_{r_1}$.  Consider an arbitrary resource $r \in R_e$. If $v_r \in (0, \frac{\lambda}{2})$, 
	\[
		\frac{z^*_r}{v_r} = \frac{3\lambda}{2\lambda + v_r}  \leq \frac{3\lambda}{2\lambda + v_{r_0}}.
	\]
	If $v_r \in [\frac{\lambda}{2}, \frac{3}{4}\lambda)$,
	\[
		\frac{z^*_r}{v_r} \leq \frac{3\lambda}{3\lambda - v_r} =  \frac{3\lambda}{2\lambda + (\lambda - v_r)} < \frac{3\lambda}{2\lambda + v_{r_0}}.
	\]
	If $v_r \in (\frac{3}{4}\lambda, \lambda)$,
	\[
		\frac{z^*_r}{v_r} = \frac{\lambda}{v_r} < \frac{3\lambda}{3\lambda - v_r} =  \frac{3\lambda}{2\lambda + (\lambda - v_r)} \leq \frac{3\lambda}{2\lambda + v_{r_0}}.
	\]
	In summary, for every $r \in R_e$, 
	\[
		\frac{z^*_r}{v_r} \leq \frac{3\lambda}{2\lambda + v_{r_0}}.
	\] 
	Since $e$ is $\lambda$-minimal, 
	\[
		\sum_{r \in R_e}v^*_r < \lambda + v_{r_0}.
	\]
	Combining these two facts, we obtain
	\begin{align*}
		\sum_{r\in R_e} z^*_r &\leq \frac{3\lambda}{2\lambda + v_{r_0}}\sum_{r \in R_e}v^*_r\\
							&<\frac{3\lambda}{2\lambda + v_{r_0}}\left(\lambda + v_{r_0}\right)\\
							&= \frac{3\lambda}{2\lambda + v_{r_0}}\left(\lambda + \frac{1}{2}v_{r_0}\right) + \frac{3\lambda}{2\lambda + v_{r_0}}\frac{1}{2} v_{r_0}\\
							&\leq  \frac{3\lambda}{2} + \frac{z^*_{r_0}}{2}.
	\end{align*}
\end{claimproof}

Now we are ready to prove Claim~\ref{cl:gap-dual-positive}.
\begin{claimproof}[Proof of Claim~\ref{cl:gap-dual-positive}]
	We shall prove that $\sum_{p\in P}y^*_p - \sum_{r\in R}z^*_r > 0$. By our construction of the dual solution, it is easy to see that
	\[
		\sum_{p\in P}y^*_p - \sum_{r\in R}z^*_r  = \left(\sum_{p\in P^+}y^*_p - \sum_{r\in R^+_f}z^*_r \right) - \sum_{r \in R(\Sigma)}z^*_r.
	\]

	Consider $\sum_{p\in P^+}y^*_p - \sum_{r\in R^+_f}z^*_r$. For every player $p \in P^+$ and every fat resource $r_f \in R^+_f$, $y^*_p$ and $z^*_{r_f}$ have the same value $1 - \frac{21}{26}\lambda$.  Therefore,
	\[
		\sum_{p\in P^+}y^*_p - \sum_{r\in R^+_f}z^*_r = \left(1 - \frac{21}{26}\lambda \right)\left(|P^+| - |R^+_f|\right).
	\] 
	We will derive a lower bound for $|P^+| - |R^+_f|$.  The player in $B_1$ is $p_0$, and $p_0$ is not matched by $M$.  Players in other $B_i$'s are covered by ${\cal E}$, so they are not matched by $M$.  Therefore, no player in $B_{ \leq \ell}$ is matched by $M$. Every fat resource that is reachable in $G_M$ from some player in $B_{ \leq \ell}$ must be matched by $M$; otherwise, such the path between the fat resource and the player would be an augmenting path with respect to $M$, contradicting the fact that $M$ is a maximum matching of $G$.  Hence, for every fat resource $r_f \in R^+_f$, $r_f$ must have an out-going edge to some player $p$ in $G_M$.  Since $r_f$ is reachable from $B_{\leq \ell}$ in $G_M$, $p$ is also reachable from $B_{\leq \ell}$ in $G_M$.  So $p \in P^+$.   We charge $r_f$ to $p$.  Every player in $P$ has in-degree at most $1$ in $G_M$, so it is charged at most once.  Recall that $B_{\leq \ell} \subseteq P^+$.  Players in $B_{\leq \ell}$ do not have incoming edges in $G_M$ because they are not matched by $M$. So players in $B_{\leq \ell}$ do not get charged.  In conclusion, 
	\[
		|P^+| - |R^+_f| \geq |B_{\leq \ell}|.
	\]
	Putting things together, we get 
	\begin{equation}
		\sum_{p\in P^+}y^*_p - \sum_{r\in R^+_f}z^*_r \geq \left(1 - \frac{21}{26}\lambda\right)|B_{\leq \ell}|.
		\label{eq:gap-lower-a}
	\end{equation}
	
	Next we bound $\sum_{r \in R(\Sigma)}z^*_r$.  As in Figure~\ref{fig:zr}(a), $z^*_r$ does not decrease as $v_r$ increases. Hence, Claim~\ref{cl:edge-dual-value} implies that for any $e \in R(\Sigma)$ and for any resource $r' \in R_e$, 
	\begin{equation}
				\sum_{r\in R_e} z^*_r \leq \frac{3\lambda}{2} + \frac{z^*_{r'}}{2}.
				\label{eq:edge-dual-value}
	\end{equation}
	Now consider a tuple $(a_i, {\cal B}_i)$. By our assumption that every addable edge is blocked, $|{\cal B}_i| \geq 1$. Let $R_i$ be the set of thin resources covered by $\{a_i\} \cup {\cal B}_i$.  Consider a resource $r \in R_i$.  If $r$ is covered by only one edge in $\{a_i\} \cup {\cal B}_i$, then charge $z_r^*$ to that edge.  If $r$ is covered by two edges in $\{a_i\} \cup {\cal B}_i$, then charge half of $z_r^*$ to one edge and the other half to the other edge.  Take any edge $e \in \{a_i\} \cup {\cal B}_i$.  Because edges in ${\cal B}_i$ are blocking edges of $a_i$, $e$ must share some resource, say $r'$, with another edge in $\{a_i\} \cup {\cal B}_i$. We conclude from \eqref{eq:edge-dual-value} that the total charge on $e$ is at most
	\[
	\sum_{r \in R_e} z_r^* - \frac{z_{r'}^*}{2} \leq \frac{3\lambda}{2}.
	\]
	Taking sum over all edges in $\{a_i\} \cup {\cal B}_i$ gives
	\[
		\sum_{r \in R_i}z^*_r \leq \left|\{a_i\} \cup {\cal B}_i\right| \cdot \frac{3\lambda}{2} =  \frac{3\lambda}{2}\left(|B_i|+1\right) \leq 3\lambda |B_i|.
	\]
	The last inequality follows from the fact that $|B_i| \geq 1$.  Recall that $(a_1,{\cal B}_1) = (\mathit{null}, (p_0, \emptyset))$ covers no resource.  Then, taking sum over $i\in [2, \ell]$ gives
	\begin{equation}
		\label{eq:gap-upper-a}
		\sum_{r \in R(\Sigma)}z^*_r = \sum_{i = 2}^{\ell} \sum_{r\in R_i}z^*_r \leq 3\lambda\left(|B_{\leq \ell}| - 1\right).
	\end{equation}

	Combining \eqref{eq:gap-lower-a} and \eqref{eq:gap-upper-a}, we obtain
	\begin{align*}
		\sum_{p\in P}y^*_p - \sum_{r\in R}z^*_r  
		&= \sum_{p\in P^+}y^*_p - \sum_{r\in R^+_f}z^*_r - \sum_{r \in R(\Sigma)}z^*_r\\
		&\geq \left(1 - \frac{21}{26}\lambda\right)|B_{\leq \ell}| - 3\lambda\left(|B_{\leq \ell}|-1\right)\\
		& = 3\lambda + \left(1 - \frac{99}{26}\lambda\right)|B_{\leq \ell}|.
	\end{align*}
	Given $\lambda = \frac{26}{99}$,
	\[
		\sum_{p\in P}y^*_p - \sum_{r\in R}z^*_r \geq 3\lambda > 0.
	\]
\end{claimproof}

We have completed the proof of Lemma~\ref{lem:gap-nonstuck}, which says that the algorithm is always able to either find an unblocked edge and call {\sc Contract} to shrink $\Sigma$ or call {\sc Build} to append an addable edge and its blocking edges to $\Sigma$.  It remains to show that the local search algorithm will make $\Sigma$ empty in finite time (i.e., satisfies $p_0$ eventually).  This analysis has been given before in~\cite{AFS12}.  We repeat it here for completeness.

\begin{lemma}
\label{lem:gap-terminate}
The algorithm terminates after a finite number of calls of {\sc Build} and {\sc Contract}.
\end{lemma}
\begin{proof}
	Define a signature vector $(|B_1|, |B_2|, \ldots, |B_{\ell}|, \infty)$ with respect to the sequence of tuples.  All ${\cal B}_1, \ldots, {\cal B}_\ell$ are mutually disjoint subsets of $\cal E$, so $|B_1| + \cdots + |B_\ell| \leq n$.  Therefore, the number of distinct signature vectors is at most $n^n$.  The signature vector evolves as the sequence is updated by the algorithm.  After each invocation of {\sc Build}, the signature vector decreases lexicographically because it gains a new second to last entry.  After each invocation of {\sc Contract}, the signature decreases lexicographically because it becomes shorter, and the second to last entry decreases by at least $1$.  Therefore, no signature vector is repeated.  As a result, the algorithm terminates after at most $n^n$ invocations of {\sc Build} and {\sc Contract}.
\end{proof}

By Lemma~\ref{lem:gap-terminate}, we can obtain an allocation in which every player receives at least $\lambda = \frac{26}{99}$ worth of resources.  This completes the proof of Theorem~\ref{thm:gap}.

\end{document}